\documentclass[11pt]{article}
\date{}
\usepackage{amsfonts}
\usepackage{amsmath, enumitem}
\usepackage{graphicx, color}
\definecolor{blue}{rgb}{0.100000,0.100000,0.701961}
\usepackage[a4paper]{geometry}
\usepackage{amsmath}
\usepackage{amssymb}
\usepackage{amsopn,ntheorem}
\usepackage[font=footnotesize]{caption}
\usepackage{subcaption}

\newcommand{\fordagger}{\mathsf{T}} 
 
\newcommand{\tr}{\mathtt{Tr}}
\newcommand{\X}{\textbf{X}}
\newcommand{\Y}{\textbf{Y}}
\newcommand{\Z}{\textbf{Z}}
\newcommand{\C}{\textbf{C}}

\newcommand{\y}{\textbf{y}}
\newcommand{\x}{\textbf{x}}

\newcommand{\diag}{\mathsf{Diag}}
\newcommand{\Del}{\mathbf{\Delta}}
\newcommand{\Xc}{\textsf{\textbf{X}}} 
\newcommand{\Yc}{\textsf{\textbf{Y}}} 
\newcommand{\Zc}{\textsf{\textbf{Z}}} 
\newcommand{\Sconc}{\textbf{\textsf{S}}} 
\newcommand{\Delc}{\boldsymbol{\mathsf{\Delta}}} 
\newcommand{\Ztilde}{\boldsymbol{\mathsf{\tilde Z}}} 
\newcommand{\Qa}{\mathsf{Q}_\mathsf{a}}
\newcommand{\Qb}{\mathsf{Q}_\mathsf{b}}
\newcommand{\Da}{\mathsf{D}_\mathsf{a}}
\newcommand{\Db}{\mathsf{D}_\mathsf{b}}
\newcommand{\aaa}{\mathsf{a}}
\newcommand{\bbb}{\mathsf{b}}

\newtheorem{theorem}{Theorem}
\newtheorem{lemma}{Lemma}

\newtheorem{definition}{Definition}

\newtheorem{remark}{Remark}

\newenvironment{proof}{{\noindent{\bf Proof:}}}{$\hfill\Box$}
\allowdisplaybreaks
\begin{document}
\title{Sampling and Distortion Tradeoffs for Indirect Source Retrieval }
\author{Elaheh Mohammadi, Alireza Fallah and Farokh Marvasti\\\small Advanced Communications Research Institute (ACRI) \\ \small Department of Electrical Engineering\\\small
Sharif University of Technology \footnote{This paper has been presented in part at GlobalSIP 2016.}}
\maketitle

\begin{abstract} 

Consider a continuous signal that cannot be observed directly. Instead, one has access to multiple corrupted versions of the signal. The available corrupted signals are correlated because they carry information about the common remote signal. The goal is to reconstruct the original signal from the data collected from its corrupted versions. Known as the indirect or remote reconstruction problem, it has been mainly studied in the literature from an information theoretic perspective. A variant of this problem for a class of Gaussian signals, known as the ``Gaussian CEO problem", has received particular attention; for example, it has been shown that the problem of  recovering the remote signal is equivalent with the problem of  recovering the set of corrupted  signals ({separation principle}). 

The information theoretic formulation of the  remote reconstruction problem assumes that the corrupted signals are uniformly sampled and the focus is on optimal compression of the samples. On the other hand, in this paper we revisit this problem from a sampling perspective. More specifically, assuming restrictions on the sampling rate from each corrupted signal, we look at the problem of finding the best sampling locations for each signal to minimize the total reconstruction distortion of the remote signal. In finding the sampling locations, one can take advantage of the correlation among the corrupted signals. The statistical model of the original signal and its corrupted versions adopted in this paper is similar to the one considered for the Gaussian CEO problem; \emph{i.e.,} we restrict to  a class of Gaussian signals. 

Our main contribution is a fundamental lower bound on the reconstruction distortion for any arbitrary nonuniform sampling strategy. This lower bound is valid for any sampling rate. Furthermore, it is tight and matches the optimal reconstruction distortion in low and high sampling rates. Moreover, it is shown that in the low sampling rate region, it is optimal to use a certain nonuniform sampling scheme on all the signals. On the other hand, in the high sampling rate region, it is optimal to uniformly sample all the signals. 
 We also consider the problem of finding the optimal sampling locations to recover the set of corrupted  signals, rather than the  remote signal. Unlike the information theoretic formulation of the problem in which these two problems were equivalent, we show that they are not equivalent in our setting.

\end{abstract}

\section{Introduction}
  In many applications, such as monitoring or sensing systems, one may be interested in reconstructing a  stochastic  source $S(t)$ that is not directly observable. Instead,  access is provided to $k$ correlated  signals $S_1(t), S_2(t), \cdots, S_k(t)$ that are corrupted  versions of $S(t)$. The goal is to reconstruct $S(t)$ from limited information that one can obtain from $S_i(t), i=1,2,\cdots, k$. While the source coding aspect of the problem is classical in information theory (see for instance \cite[Sec. 3.5]{Berger}), its signal processing and sampling aspect has not received much attention. In this paper, we study the sampling aspect of this problem for \emph{stochastic} signal $S(t)$ and the $k$ corrupted  versions $S_i(t)$. We take into account the fact that the  correlation among the signals $S_i(t), i=1,2,\cdots, k$, can help decrease the sampling rate or improve the signal reconstruction accuracy.

\subsection*{A. System model} 
It is known that any \emph{deterministic} continuous function $s(t)$ defined on the interval $[0,T]$ can be expressed in terms of sinusoids as follows:
$$s(t)=\sum_{\ell=0}^{\infty}[a_{\ell}\cos(\ell\omega  t)+b_{\ell}\sin(\ell\omega  t)], \qquad t\in [0,T],$$
where $\omega =2\pi/{T}$. If the number of non-zero coefficients $a_{\ell}$ and $b_{\ell}$ are limited, the signal $s(t)$ is sparse in the frequency domain.
Herein, we  consider  a \emph{stochastic} signal $S(t)$, and show its Fourier coefficients by random variables $A_{\ell}$ and $B_{\ell}$. We assume that the coefficients $A_{\ell}$ and $B_{\ell}$ are zero when $\ell>N_2$ or $\ell<N_1$ for some natural numbers $N_1\leq N_2$, \emph{i.e.,} $S(t)$ is a bandpass stochastic signal:
\begin{align}
S(t)&=\sum_{\ell=N_1}^{N_2}[A_{\ell}\cos(\ell\omega  t)+B_{\ell}\sin(\ell\omega  t)], \qquad t\in [0,T],\label{eqdefs0}\end{align}
where the coefficients $A_{\ell}$ and $B_{\ell}$ for $N_1\leq \ell\leq N_2$ are mutually independent identically distributed (i.i.d.) normal $\mathcal{N}(0, 1)$ variables, \emph{i.e.,} the original signal is white and Gaussian. We cannot observe $S(t)$ directly. Instead,  we have $S_1(t), S_2(t), \cdots, S_k(t)$, also defined on $t\in[0,T]$, that are corrupted  versions of $S(t)$. The corrupted versions of the signal can be expressed as
\begin{align}
S_i(t)&=\sum_{\ell=N_1}^{N_2}[A_{i\ell}\cos(\ell\omega  t)+B_{i\ell}\sin(\ell\omega  t)], \qquad t\in [0,T],~~i\in \{1,2,\cdots, k\},\label{eqdefs}
\end{align}
where  $A_{i\ell}=A_\ell+W_{i\ell}$ and $B_{i\ell}=B_\ell+V_{i\ell}$; here
 $W_{i\ell}$ and $V_{i\ell}$ are independent perturbations that are added to the original signal.  It is assumed that the perturbations $W_{i\ell}$ and $V_{i\ell}$ for $i\in \{1,2,\cdots, k\}$ and $\ell\in \{N_1, N_1+1, \cdots, N_2\}$ are i.i.d.\ variables according to  $\mathcal{N}(0, \eta)$. The perturbations are also mutually independent of the signal coefficients $A_\ell$ and $B_\ell$ for $N_1\leq \ell\leq N_2$.  
The statistical model assumed for the coefficients $A_{\ell}$, $B_{\ell}$, $V_{i\ell}$, $W_{i\ell}$ parallels the one for  the ``Gaussian CEO problem"  \cite{CEO-Berger, Gaussian-CEO}. A summary of the model parameters is given in Table~\ref{table:summary}.

 We are allowed to take $m_i$ samples from the $i$th corrupted  signal $S_i(\cdot)$ at time instances $t_{i1}, t_{i2}, \cdots, t_{im_i}\in[0,T]$ of our choice, for $i=1,2,\cdots, k$. Therefore $m_i/T$ can be viewed as the sampling rate of the $S_i(\cdot)$. We assume that the samples are noisy. The sampling noise can model quantization noise of an A/D converter, or the noise incurred by transmitting the samples to a fusion center over a communication channel.  The sampling noise of each signal $S_i(t)$ is modeled by an independent zero-mean Gaussian random variable with variance $\sigma_i^2$.
We use the samples to reconstruct either the remote signal $S(t)$, or the collection of corrupted  signals $S_1(t), S_2(t), \cdots, S_k(t)$. The motivation for reconstructing $\{S_i(t), i=1,\cdots, k\}$ is twofold: firstly, this would parallel the literature on indirect source coding, where the reconstruction distortion  of  intermediate signals is shown to be equivalent with that of the original signal (the separation theorem \cite{Witsenhausen, Wolf, Dobrushin}). Secondly, these individual signals $S_i(t)$ may contain some other information of interest besides $S(t)$, e.g., the differences $S(t)-S_i(t)$ might be correlated with some other signal of interest.

The reconstruction of the  remote signal $S(t)$ and the corrupted signal $S_i(t)$  are denoted  by  $\hat{S} (t)$ and $\hat{S}_i (t)$, respectively. These reconstructions are calculated using the Minimum Mean Square Error (MMSE) criterion, \emph{i.e.,} $\hat{S} (t)$ is the conditional expectation of the $S(t)$ given all the samples. 
The goal is to optimize over the sampling times $t_{ij}$ to minimize the distance between the signals and their reconstructions. More specifically, we consider the minimization
\begin{align}\mathsf{D}_{\mathsf a\min}=\min_{\{t_{i1},t_{i2}, \cdots, t_{im_i}\}_{i=1}^k}\frac{1}{T}\int_{t=0}^T\mathbb{E}\{|\hat{S}(t)-S(t)|^2\}dt,\label{eqn:dist1am}\end{align}
for the remote signal, or the minimization 
\begin{align}\mathsf{D}_{\mathsf b\min}=\min_{\{t_{i1},t_{i2}, \cdots, t_{im_i}\}_{i=1}^k}\frac{1}{T}\int_{t=0}^T\sum_{i = 1}^{k}\mathbb{E}\{|\hat{S}_i(t)-S_i(t)|^2\}dt.\label{eqn:dist1bm}\end{align}
for reconstruction of  the $k$ corrupted  signals $S_i(t)$, $i=1,2,\cdots, k$. Here  $t_{ij}$ is the   $j$th sampling time of the $i$th signal.

\begin{table*}
\begin{center}
\begin{small}
\begin{tabular}{|c|c|}
\hline
Notation & Description \\
\hline
$T$, $\omega$ & $T$ is signal period and $\omega=2\pi/T$ \\
\hline
$k$& Number of corrupted signals\\
\hline
$S(t)$ and $S_i(t)$& The original signal and the $i$th corrupted signal, respectively.\\
\hline
$A_\ell, B_\ell$& Fourier series coefficients of the original signal $S(t)$\\ &$\ell\in [N_1:N_2]$, $A_\ell, B_\ell\sim\mathcal{N}(0,1)$.\\
\hline
$A_{i\ell}, B_{i\ell}$& Fourier series coefficients of the $i$th corrupted signal $S_i(t)$\\ &$\ell\in [N_1:N_2]$, $A_{i\ell}, B_{i\ell}\sim\mathcal{N}(0,1+\eta)$\\
\hline
$\eta$&Variance of the perturbation added to $A_\ell$ 
\\&to produce $A_{i\ell}$ for $1\leq i\leq k$.
\\
\hline
 & The support of the input signal in frequency
\\$N$, $N_1$ and $N_2$& domain is from $N_1\omega$ to $N_2\omega$. 
\\&$N=N_2-N_1+1$.  \\&$2N$ is the number of free variables of each signal\\
\hline
$m_i$ & Number of noisy samples of the $i$th corrupted signal\\
\hline
$\sigma_i^2$ & Variance of the sampling noise of the $i$th corrupted signal\\
\hline
$\{t_{i1}, t_{i2}, \cdots, t_{im_i}\}$ & Sampling time instances of the $i$th corrupted signal \\
\hline
$\phi_i=\frac{m_i}{2\sigma_i^2}+\frac{1}{\eta}$, $\Phi_p=\sum_{i=1}^k \phi_i^p$ & Definitions of $\phi_i$ and $\Phi_p$ used in Theorems \ref{Thm1} and \ref{Thm2}\\
\hline
\end{tabular}
\end{small}
\end{center}
\caption{Definition of main parameters.}
\label{table:summary}
\end{table*}

\subsection*{B. Related works} The problem that we defined above is novel. However, it relates to the literature on \emph{remote signal reconstruction} and \emph{distributed sampling}. The former has been only studied from an information theoretic and source coding perspective, while the latter has been mainly considered in the context of compressed sensing and wireless sensor networks. Finally, while we consider sampling of multiple signals, there are some previous works that study sampling rate and reconstruction distortion of a single source, \emph{e.g.} see \cite{kipnis1, timestampless, Boda, firstpaper}.

\emph{Remote signal reconstruction:} Reconstruction distortion of correlated signals (lossy reconstruction) is a major theme in multi-user information theory for the class of discrete \emph{i.i.d.} signals. However, the emphasis in multi-user source coding is generally on the quantization and compression rates of the sources, and not on the sampling rates. It is assumed that the signals are all uniformly sampled at the Nyquist rate. The indirect source coding problem  was first introduced by Dobrushin and Tsyabakov  in information theory literature \cite{Dobrushin}. This work and subsequent information theoretic ones deal with discrete sources, by assuming that we have several bandlimited signals $S_i(t)$, that are sampled at the Nyquist rate  with no distortion at the sampling phase. Then, assuming a finite quantization rate for storing the samples, the task is to minimize the total reconstruction distortion (which is only due to quantization). On the other hand, our work in this paper is on indirect source \emph{retrieval} and not indirect source \emph{coding}, as we do not study the quantization aspect of the problem. Rather, we focus on the distortion incurred by the sampling rate (which can be below the Nyquist rate), and the additive noise on the samples. We also allow  nonuniform sampling to decrease the distortion.

\emph{Distributed sampling:}
From another perspective, our problem relates to the {distributed sampling} literature. In distributed compressed sensing, the structure of correlation among multiple signals is their joint-sparsity. This problem was studied in \cite{Baron}, where signal recovery algorithms using linear equations obtained by distributed sensors were given. Authors in \cite{Hormati} model the correlation of two signals by assuming that one is related to the other by an unknown \emph{sparse} filtering operation. The problem of centralized reconstruction of two correlated signals based on their distributed samples is studied, and its similarities with the Slepian-Wolf theorem in information theoretic distributed compression are pointed out. Motivated by an application in array signal processing, the authors in \cite{Ahmed} consider signal recovery for a specific type of correlated signals, assuming that the signals lie
 in an unknown but low dimensional linear subspace.

Spatio-temporal correlation of the distributed signals is a significant feature of wireless sensor networks and can be utilized for sampling and  data collection  \cite{Ganesan}. In \cite{Bandyopadhyay}, the spatio-temporal sampling rate tradeoffs of a sensor network for minimum energy usage is studied. Authors in \cite{Zordan} provide a mathematical model for the spatio-temporal correlation of the signals observed by the sensor nodes. In \cite{Masiero}, the spatio-temporal statistics of the distributed signals are used by the Principal
Component Analysis (PCA) method to find transformations
that sparsify the signal. Compressed sensing is then used for signal recovery. In  \cite{Bajwa},  a compressive
wireless sensing is given for signal retrieval at a fusion center from an ensemble of spatially distributed sensor nodes. See also \cite{Wang} for a distributed algorithm based on sparse random
projections for signal recovery in sensor networks.

\subsection*{C. Our contributions}  Having chosen a particular sampling strategy (such as uniform sampling), one obtains a value for its reconstruction distortion. This value serves as an \emph{upper bound} on the optimal reconstruction distortion. But we are interested to know how close we come to the optimum distortion with this particular sampling strategy. To estimate this, it is desirable to find \emph{fundamental lower bounds} on the reconstruction distortion which hold regardless of the sampling strategy. In this paper, we provide such a lower bound  \emph{for any arbitrary sampling rates} (i.e., any arbitrary values for $m_i$) for both of the problems of reconstructing the remote signal, or the collection of the corrupted  signals. Furthermore, this lower bound is shown to  coincide with the optimal distortion 
 in the high and low sampling regions. In other words, while our fundamental lower bound applies to any arbitrary sampling rates (and provides information about the general behavior of the optimal distortion), it is of particular interest in the high and low sampling regions for which it becomes tight.

\textbf{High sampling rate region:} this refers to the case of $m_i>2N_2$ for all $i$, \emph{i.e.,} we are sampling each of the signals above the Nyquist rate.  Our result in the high sampling region, it is optimal to use uniform sampling for each signal, and the lower bound matches the distortion yielded by uniform sampling.
From a practical perspective, the high sampling rate region is relevant to the case of \emph{high sampling noise} (large $\sigma_i$). When each of the samples taken from a signal is very noisy, we desire to oversample. As an example, if the sampling noise is modeling the quantization noise of an A/D converter, we can consider a $\Sigma\Delta$ modulator that oversamples with high quantization noise. If the sampling noise models the channel noise incurred by transmitting the samples to a fusion center over a wireless medium, then oversampling provides a redundancy to combat the channel noise.

\textbf{Low sampling rate region:} The low sampling rate region refers  to the case of $\sum_{i=1}^k m_i\leq N$. Our result for the low sampling rate region states that a certain nonuniform sampling strategy is distortion optimal, and the lower bound matches the distortion yielded by this nonuniform sampling. In progressive or multi resolution applications, one wishes to recover a low-resolution version of a signal, and based on that decide whether to seek a higher resolution version. The low sampling rate region can be helpful in the low-resolution phase. We also make the following comments about the low-sampling rate region:
\begin{itemize}
\item Our result allows us to quantify the difference between the following two cases: (i) not taking any samples at all, and (ii) taking a total of $N$ samples $\sum_{i=1}^k m_i\leq N$. By showing that the optimal distortion in case (ii) differs from the optimal distortion in case (i) by at most 3dB, we conclude the following negative result: there is no gain beyond 3dB by using any, however complicated, {nonuniform} sampling strategy.\footnote{However, one should also note in sensitive applications, such as radar, extra sampling to improve the resolution by 3dB can be valuable.}
\item If a limitation on the number of samples that we can possibly take is enforced on us as a physical constraint, it is still of interest to know the best possible achievable distortion. 
\end{itemize}

Finally, we comment on a difference between the low and high sampling rate regions. Suppose that we have a total budget on the number of samples $m=\sum_{i=1}^k m_i$ that we can take from the corrupted  signals. Then,
(i) for low sampling rates, $m\leq N$:  if we want to reconstruct either the \emph{remote signal} or  the collection of \emph{corrupted  signals}, it is best to take the samples from the signal with the smallest sampling noise, \emph{i.e.,} if $\sigma_1\leq \sigma_i$ for all $i$, it is optimal to choose $m_1=m, m_i=0$ for $i>2$. (ii) for high sampling rates:  for reconstructing the \emph{remote signal}, taking more samples from the less noisy signals is advantageous, but to reconstruct the  collection of \emph{corrupted  signals},  it is no longer true that we should take as many sample as possible from the less noisy signal. The optimum number of samples that we should take from each signal is an optimization problem, with a solution depending on the parameters of the problem.

\subsection*{D. Proof Techniques}

 The Gaussian assumption implies that the MMSE and linear MMSE (LMMSE) are identical. Since LMMSE estimator only depends on the second moments, the problem reduces to a linear algebra optimization problem. However, this optimization problem is not easy because the variables we are optimizing over, are sampling locations $t_{ij}$ that show up as arguments of sine and cosine functions. Sine and cosine functions are nonlinear, albeit structured, functions. The goal would be to exploit their structure to solve the optimization problem. 

To find fundamental lowers bounds for the reconstruction distortion, we utilize various matrix inequalities: these include (i) an inequality in {majorization theory} that relates trace of a function of a matrix to the diagonal entries of the matrix (see \cite[Chapter 2]{Bhatia}),  (ii) the matrix version of the arithmetic and harmonic means inequality (iii) the L\"{o}wner-Heinz theorem for operator convex functions, and (iv) more importantly a new reverse majorization inequality (Theorem \ref{Lemmaab}). A contribution of this paper is this reverse majorization inequality that might be of independent interest. Majorization inequalities state that the diagonal entries of a Hermitian matrix $F$ are majorized by the eigenvalues of $F$. Therefore, if $F_{\diag}$ is a diagonal matrix wherein we have kept the diagonal entries of $F$ and set the off-diagonals to zero, we will have
\begin{align} \tr \left[ F^{-1}\right ]& \geq \tr \left[ F_{\diag}^{-1}\right ]
\end{align}
Our reverse majorization inequality goes in the reverse direction. For certain matrices $F$ and $G$ of our interest, we show that 
\begin{align} \tr \left[F^{-1}G\right ]& \leq \tr \left[ F_{\diag}^{-1}G_{\diag}\right ].
\end{align}

\subsection*{E. Notation and Organization} 
Uppercase letters are used for random variables and matrices, whereas lowercase letters show (non-random) values. Vectors are denoted by  lowercase bold letters  (such as $\x$), and random vectors are denoted by either uppercase bold  letters  (such as $\X$) or bold sans-serif letters  (such as $\Xc$). The covariance of a random vector $\X$ is denoted by $C_{\X}$. Given a matrix $A$,  $A_{\diag}$  denotes the matrix formed by keeping the diagonal entries of $A$ and changing the off-diagonal entries to zero. Given a vector $\x$, $\diag(\x)$ denotes the diagonal matrix where its diagonal entries are coordinates of $\x$. The symbol $\bigoplus$ is used for the direct sum and $\otimes$ is used for the Kronecker product of matrices. We write $A\leq B$ if $B-A$ is positive semi-definite. For a Hermitian matrix $A$ with eigen decomposition $PDP^{-1}$ and real function $f$, $f(A)$ is defined as $Pf(D)P^{-1}$ in which $f(D)$ is a diagonal matrix, where function $f$ is applied to the diagonal entries of $D$.

The paper is organized as follows: In Section \ref{sec:main-results},  the main results of the paper are presented. In Section \ref{sec:problemformulation}, the problem formulation is derived.  In  Section   \ref{proofs}, the proofs of the main results  are given. Finally,  in Appendices  \ref{sec:prelim} and  \ref{app}, the mathematical tools and technical details which have been used in the main proofs are given . In Appendix \ref{app}, a new reverse majorization theorem is derived which might be of independent  interest in linear algebra.

\section{Main Results} \label{sec:main-results}
Let $N=N_2-N_1+1$. We make the following definitions:
\begin{itemize}
\item We call $(N_2-N_1+1)f_0=Nf_0$ the \emph{signal bandwidth} (of the bandlimited signal), where $f_0=1/T$. 
\item We call $2N_2f_0$ the \emph{Nyquist rate} (twice the maximum frequency of the signal). 
\item We call $m_i/T=m_if_0$ the \emph{sampling rate} of the $i$-th corrupted signal. It is the number of total samples from $S_i(t)$ in $[0,T]$, divided by  period length $T$. If we periodically extend the $m_i$ samples (periodic nonuniform sampling), $m_i/T$ will be the number of samples taken {per unit time}, hence called the sampling rate. 
\end{itemize}

To state the main result, we need a definition. For any real $p$, let $\phi_i=\frac{m_i}{2\sigma_i^2}+\frac{1}{\eta}$, and \begin{align}\Phi_p=\sum_{i=1}^k \phi_i^p.\label{defPhi}\end{align}

\begin{theorem} 
[\textbf{Reconstruction of the original signal}] The following general lower bound on the optimal distortion (given in \eqref{eqn:dist1am}) holds for any given sampling rates:
\begin{align}
\mathsf{D}_{\mathsf a\min} &\geq \max\left(N - N\sum_{i=1}^{k} \frac{m_i}{2(1+\eta)N+2\sigma_i^2}~,~ \frac{N}{(1+\frac{k}{\eta})-\frac{k^2}{\eta^2} (\Phi_1)^{-1} }\right),\label{eqn:lowerboundOriginalSignal}
\end{align}
where $\Phi_1$ is defined as in \eqref{defPhi}. Furthermore, the lower bound given in \eqref{eqn:lowerboundOriginalSignal} is tight (the inequality is an equality)  in the following cases:
\begin{itemize}
\item when $\sum_{i=1}^k m_i\leq N$:  in this case, the optimal sampling points, $t_{ij}$, are all distinct for $1\leq i\leq k, 1\leq j\leq m_i$, and belong to the set $\{0,T/N, \cdots, (N-1)T/N \}$. 
\item when $m_i > 2N_2$ for $1\leq i\leq k$: in this case uniform sampling of each signal $S_i(t)$ is optimal.
\end{itemize}\label{Thm1}
\end{theorem}

\begin{remark} From optimality of the lower bound for large values of $m_i$, we obtain that 
$$\lim_{\forall i: m_i\rightarrow \infty}\mathsf{D}_{\mathsf a\min} = \frac{N}{1+\frac{k}{\eta}}>0,$$
is strictly positive. The reason is that $m_i=\infty$ implies full access to the set of corrupted signals, $\{S_i(t)\}$, but even in this case we cannot perfectly reconstruct $S(t)$ if the corruption variance $\eta>0$.
\end{remark}

\begin{proof} To prove the theorem, it suffices to show that for any arbitrary choice of sampling time instances, we have
\begin{align}
\mathsf{D}_{\mathsf a} &\geq N - N\sum_{i=1}^{k} \frac{m_i}{2(1+\eta)N+2\sigma_i^2}\label{Lower1S}
\end{align}
and moreover, equality in the above equation holds if $\sum_{i=1}^k m_i\leq N$. This claim is shown in Section \ref{proofsLower1S}; in this section, the optimality of choosing $t_{ij}$ from $\{0,T/N, \cdots, (N-1)T/N \}$ is also established. Next, we also show that 
\begin{align}
\mathsf{D}_{\mathsf a\min} &\geq \frac{N}{(1+\frac{k}{\eta})-\frac{k^2}{\eta^2} (\Phi_1)^{-1} },\label{Lower2S}
\end{align}
and furthermore equality in the above equation holds if $m_i>2N_2$. The proof of this claim is given in Section \ref{proofsLower2S}; in this section, the optimality of uniform sampling is also established. This completes the proof.
\end{proof}

\begin{theorem}[\textbf{Reconstruction of the set of corrupted signals}] The following general lower bound on the optimal distortion (given in \eqref{eqn:dist1bm}) holds:
\begin{align}
\mathsf{D}_{ \mathsf b\min} \geq \max\left(Nk(1+\eta) -N\left((1+\eta)^2+(k-1)\right)\sum_{i=1}^{k} \frac{ m_i}{2(1+\eta)N+2\sigma_i^2}~,~N\Phi_{-1}+\frac{N}{\eta(\eta+k)-\Phi_{-1}}\Phi_{-2}\right),\label{eqn:lowerboundOriginalSignal2}
\end{align}
where $\Phi_{-1}, \Phi_{-2}$ are defined as in \eqref{defPhi}. Furthermore, the lower bound given in \eqref{eqn:lowerboundOriginalSignal2} is achieved  in the  following two cases:
\begin{itemize}
\item when $m=\sum_{i=1}^k m_i\leq N$:  in this case, the optimal sampling points, $t_{ij}$, are all distinct for $1\leq i\leq k, 1\leq j\leq m_i$, and belong to the set $\{0,T/N, \cdots, (N-1)T/N \}$. 
\item when $m_i > 2N_2$ for $1\leq i\leq k$: in this case uniform sampling of each signal $S_i(t)$ is optimal.
\end{itemize}\label{Thm2}
\end{theorem}

\begin{proof} We first show that
\begin{align}
\mathsf{D}_{ \mathsf b\min} \geq Nk(1+\eta) -N\left((1+\eta)^2+(k-1)\right)\sum_{i=1}^{k} \frac{ m_i}{2(1+\eta)N+2\sigma_i^2}. \label{Lower1corrupted}
\end{align}
The proof of this claim is given in Section \ref{proofsLower1Si}. There, we also prove  that when  $\sum_{i=1}^k m_i\leq N$, the equality in the above equation holds and  the optimal time instances are distinct $t_{ij}$ chosen  from $\{0,T/N, \cdots, (N-1)T/N \}$. Next, we show that  
\begin{align}
\mathsf{D}_{\mathsf b\min} &\geq N\Phi_{-1}+\frac{N}{\eta(\eta+k)-\Phi_{-1}}\Phi_{-2},\label{Lower2corrupted}
\end{align}
and furthermore, equality in the above equation holds if $m_i>2N_2$. The proof is given in Section \ref{proofsLower2Si} where we show that  the optimal  points are uniform samples of the corrupted signals. 
\end{proof}

Finally, as a technical result, we also derive a new reverse majorization inequality, given in Theorem \ref{Lemmaab} (Appendix \ref{app}), which might be of interest in linear algebra.

In Fig.~\ref{fig1}, the lower bounds are plotted assuming that $m_i=m$ for all $i$, for 
$k=3, N_1=5, N=15, N_2=19, \sigma_i^2=1, \eta=0.1$. Also, the distortion of the uniform sampling strategy is also plotted, and serves as an upper bound for the optimal distortion curve; thus, optimal distortion curve lies in between the two curves. The two curves match when $m>2N_2=38$. Also, the lower bound is known to be tight in the low sampling rates. This part of the lower bound is drawn in color red. While Fig.~\ref{fig1} is plotted for $\sigma_i^2=1, \eta=0.1$, we observe from numerical simulation that the gap between the lower and upper bound decreases as we increase the sampling or corruption noises.

\begin{figure}[t!]
    \centering
    \begin{subfigure}[t]{0.5\textwidth}
        \centering
\includegraphics[scale=0.4,angle=0]{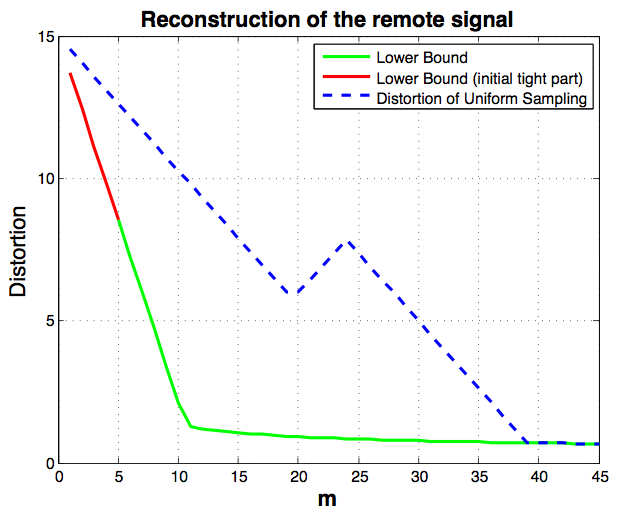}
    \end{subfigure}%
    ~ 
    \begin{subfigure}[t]{0.5\textwidth}
        \centering
\includegraphics[scale=0.4,angle=0]{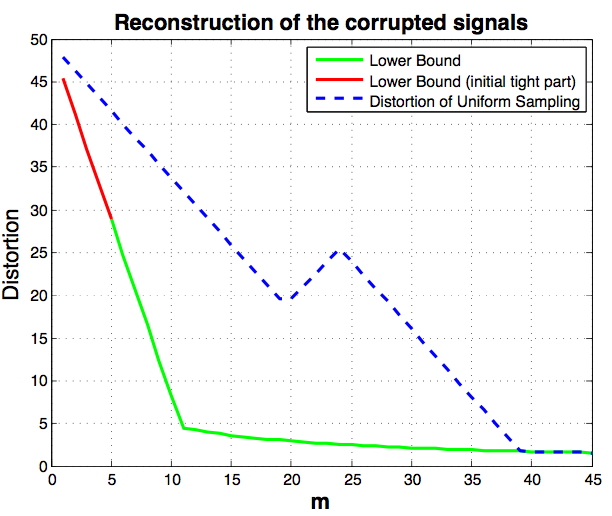}
    \end{subfigure}
    \caption{The the lower bound versus the distortion obtained by the uniform sampling strategy. It is assumed that $m_i=m$ for all $i$. The parameters for this figure are $(k, N_1, N, N_2, \sigma_i^2, \eta) = (3,5,15,19,1,0.1)$.
}
\label{fig1}
\end{figure}

\emph{Discussion 1:} We know the exact value of optimal distortion for $m_i>2N_2$. As argued in Section 1.C, this corresponds to the high sampling rate region and is practically relevant when we have high sampling noise. If the sampling noise models the quantization error of an A/D converter, this result  allows us to answer the question of  whether it is better  to collect some accurate samples from the signals, or collect many more less accurate samples from them (a problem related to selecting an appropriate $\Sigma\Delta$ modulator). 

\emph{Discussion 2:} 
Assume that the sampling noise  variances of the corrupted  signals satisfy $\sigma_1^2\leq \sigma_2^2\leq \cdots \leq \sigma_k^2$. Let us assume that we have a fixed total budget of $m$ samples that we can distribute among the $k$ signals, \emph{i.e.,} $\sum_{i=1}^k m_i=m$. Then if $m\leq N$, the lower bound will be tight. Regardless of whether we want to construct the original signal or the collection of corrupted signals, one can verify that the total distortion (subject to $\sum_{i=1}^k m_i=m$) is minimized when $m_1=m$ and $m_i=0$ for $i>1$, \emph{i.e.,} all of the samples are taken from the first signal, $S_1(t)$, which has the minimum sampling noise variance. 
\\
\noindent
However, the problems of reconstructing the {remote signal} and the set of corrupted signals are not equivalent. To see this assume that we keep the constraint $\sum_{i=1}^k m_i=m$ for some $m>2N_2k$, and further assume that $m_i > 2N_2$. For the reconstruction of the \emph{remote signal}, it can be verified that the total distortion (subject to $\sum_{i=1}^k m_i=m,~~ m_i > 2N_2$) is minimized when $m_i$ for $i \geq 2$ reaches its minimum of $2N_2+1$, which means that we take as many samples as  possible from the signal with minimum sampling noise variance $S_1(t)$.  On the other hand, if we wish to minimize the total distortion of the entire \emph{corrupted  signals}, the optimum value for $m_i$ can vary depending on the value of the parameters;  it is not necessarily true that it is better to sample from the less noisy signal. The intuitive reason for this is as follows: if sampling noise of a signal is very high, we need to take lots of samples from it to be able to have a reconstruction with low distortion. However, if the sampling noise of a signal is very low, we are able to have a good reconstruction with few samples; taking more samples yields  a negligible improvement in distortion.  Therefore, if we aim to reconstruct \emph{all} the signals (i.e., to minimize the sum of the distortions of the signals), the most economic way could be taking less  samples from the less noisy signal.  For instance, let us consider the case of two signals, $k=2$, and take the corruption noise variance to be $ \eta = 0.5$.  Assume that samples from the first signal are taken with variance  $\sigma_1 = 1$, which is less than $\sigma_2 = 10$, the noise variance of samples from the second signal.  Then
for total sample budget $m=100$, the optimum $(m_1, m_2)$ subject to $m_i > 2N_2$ is given in the following for different values of $N_2$. For $N_2=20$, the optimal choice is $(m_1, m_2) = (41, 59)$; observe that it is better to take $59$ samples from the signal with more sampling noise.  For $N_2 = 10$, the optimal choice is $(m_1, m_2) = (32, 68)$. Here $m_i > 2N_2=20$  and the optimal choice is not one of the boundary pairs $(m_1, m_2)=(21, 79)$ or $(79, 21)$.

\section{Problem Formulation}\label{sec:problemformulation}
In this section, we  state the matrix  representation of our problem. 
Let $\X$ be a column vector, consisting of the Fourier coefficients of $S(t)$
\begin{align}
\X = [{A_{N_1}, A_{N_1+1}, \cdots, A_{N_2}, B_{N_1}, \cdots, B _{N_2}}]^\fordagger,\label{Xscoefs}
\end{align}
where  $^\fordagger$ is used for the transpose operation. Similarly, $\X_i$ is a column vector, consisting of the coefficients of $S_i(t)$
\begin{align}
\X_i &= [{A_{iN_1}, A_{i(N_1+1)}, \cdots, A_{iN_2}, B_{iN_1}, \cdots, B _{iN_2}}]^\fordagger
\\&= \X+[{V_{iN_1}, \cdots, V_{iN_2}, W_{iN_1}, \cdots, W _{iN_2}}]^\fordagger
\\&= \X+\Del_i,\label{Xiscoefs}
\end{align}
where $\Del_i=[{V_{iN_1}, \cdots, V_{iN_2}, W_{iN_1}, \cdots, W _{iN_2}}]^\fordagger$.   We assume that $\X$ consists of mutually independent Gaussian random variables $\mathcal{N}(0, 1)$, \emph{i.e.,} the signal $S(t)$ is white. Therefore, $C_\X=I_{2N\times 2N}$. Moreover, the random variables $V_{ij}$ and $W_{ij}$ are assumed to be independent with  the probability distribution $\mathcal{N}(0, \eta)$ for some $\eta > 0$. 

\begin{table*}
\begin{center}
\begin{small}
\begin{tabular}{|c|c|c|}
\hline
Notation & Description & Helpful properties \\
\hline
$\mathbf X$ &  Column vector of size $2N$ of Fourier coefficients of $S(t)$ & $C_{\mathbf X}=I_{2N\times 2N}$ \\
$\mathbf X_i$ &  Column vector of size $2N$ of Fourier coefficients of $S_i(t)$&
$C_{\X_i}=(1+\eta) I$ \\&  for $i = 1, 2, \cdots, k$ & $C_{\X_i \X_j}= I$ for $i\neq j$
\\
$\Xc$ & Vectors $\mathbf X_i$ Stacked on the bottom of each other;  length $2Nk$ &
$C_{\Xc} = \Lambda \otimes I_{2N\times 2N}$
\\
\hline
$\Lambda$ & Symmetric $k\times k$ matrix given in \eqref{Rho}&\\
$\Gamma$ & $\Gamma=\Lambda^{-1}$; Explicit formula given in \eqref{eqn:Rho-1}&\\ 
\hline
$\Del_i$ & The $i$th corruption vector: $\Del_i=\mathbf X_i-\mathbf X$& $C_{\Del_i}=\eta I_{2N\times 2N}$\\
$\Delc$ & Column vectors $\Del_i$ Stacked on the bottom of each other& $C_{\C}=\eta I_{2Nk\times 2Nk}$\\
\hline
$\mathbf{{S}}_i$ & Column vector of size $m_i$; Samples of $S_i(t)$ at $t_{ij}$, $j=1,..., m_i$& $\mathbf{{S}}_i=Q_i\X_i$
\\
$Q_i$ & $m_i\times 2N$ Matrix specified by harmonics at $t_{ij}$, given in \eqref{defQi} &\\
$\Qa$ & Matrices $Q_i$ stacked on the bottom of each other; size: $(\sum m_i)\times 2N$& \\
$\Qb$ & Direct sum of matrices $Q_i$; size: $(\sum m_i)\times 2Nk$&\\
\hline
$\mathbf{Y}_i$ & Noisy samples of the $i$th signal $S_i(t)$; length $m_i$ & $\mathbf{Y}_i=\mathbf{S}_i+\mathbf{Z}_i$\\
$\Yc$& Vectors $\mathbf Y_i$ Stacked on the bottom of each other;  length $m=\sum m_i$&\\
$\mathbf{Z}_i$ & Sampling noise vector of the $i$th signal $S_i(t)$; length $m_i$&$C_{\Z_i} =\sigma_i^2 I_{m_i \times m_i}$\\
$\Zc$& Vectors $\mathbf Z_i$ Stacked on the bottom of each other;  length $m=\sum m_i$&
$C_{\Zc} =\bigoplus_{i=1}^k\sigma_i^2 I_{m_i \times m_i}$\\
$\Ztilde$& $\Ztilde$ is defined as $\Qb \mathsf{\Delc}+\Zc$. It appears in $\Yc=\Qa \mathbf{X}+\Ztilde.$ &
$C_{\Ztilde}=\eta \Qb \Qb^\fordagger +C_{\Zc}$
\\
\hline
\end{tabular}
\end{small}
\end{center}
\caption{Definition of auxiliary parameters.}
\label{table:summary2}
\end{table*}

Vectors $\X_i$ and $\X_j$ are correlated, because they are both corrupted versions of $\X$. Their cross covariance can be computed as $C_{\X_i \X_i}=(1+\eta) I$ and $C_{\X_i \X_j}= I$ for $i\neq j\in\{1,2,\cdots, k\}$. One can verify that for any $j$, the covariance matrix for the $k$ random variables $A_{ij}$ for $i=1,2,\cdots, k$ is equal to
 \begin{align}
\Lambda=\begin{pmatrix}{1+\eta} & {1 } &{1 } & \cdots & {1 }
\\{1} & {1+\eta}&{1} & \cdots&{1}
\\\vdots&\vdots&\vdots&\vdots&\vdots
\\{1} & {1} &{1} & \cdots & {1+\eta}  \end{pmatrix}_{k\times k}\label{Rho}.
\end{align}

Suppose that the $i$th  signal, $S_i(t)$, is sampled at time instances $t_{ij}$ for   $i = 1, 2 , \cdots, k$ and $j= 1,2 \cdots, m_i$. Hence,
$${S}_i(t_{ij})=\sum_{\ell=N_1}^{N_2}[{A}_{i\ell }\cos(\ell\omega  t_{ij})+{B}_{i\ell}\sin(\ell\omega  t_{ij})].$$ To represent the problem in a matrix form, we define $\mathbf{S}_i$ to be the vector of  samples as  
$$\mathbf{{S}}_i=[{S_i}(t_{i1}), {S_i}(t_{i2}), \cdots,  {S_i}(t_{im_i})]^\fordagger.$$
Therefore, we have
 $\mathbf{{S}}_i=Q_i \X_i$, where $\X_i$ is defined in \eqref{Xiscoefs} and   $Q_i$ is an $m_i\times 2(N_2-N_1+1)=m_i\times 2N$ matrix of the form
\small
\begin{align}Q_i=\begin{pmatrix}\cos(N_1\omega  t_{i1}) & \cos((N_1+1)\omega  {t_{i1}}) &\cdots&\cos(N_2\omega t_{i1})
&\sin(N_1\omega  t_{i1}) & \sin((N_1+1)\omega  t_{i1}) &\cdots&\sin(N_2\omega  t_{i1})
\\
\cos(N_1\omega  t_{i2}) & \cos((N_1+1)\omega  t_{i2}) &\cdots&\cos(N_2\omega  t_{i2})
&\sin(N_1\omega  t_{i2}) & \sin((N_1+1)\omega  t_{i2}) &\cdots&\sin(N_2\omega  t_{i2})\\
&&&\vdots&\vdots&&&\\
\cos(N_1\omega  t_{im_i}) & \cos((N_1+1)\omega  t_{im_i}) &\cdots&\cos(N_2\omega  t_{im_i})
&\sin(N_1\omega  t_{im_i}) & \sin((N_1+1)\omega  t_{im_i}) &\cdots&\sin(N_2\omega  t_{im_i})
\end{pmatrix}.\label{defQi}\end{align}\normalsize
Moreover, the observation vector for the $i$th signal is of the following form 
 $$\mathbf{Y}_i=\mathbf{S}_i+\mathbf{Z}_i=Q_i\mathbf{X}_i+\mathbf{Z}_i=Q_i\mathbf{X}+Q_i\Del_i+\mathbf{Z}_i,$$
in which  $\mathbf{Z}_i$ is the $i$th noise vector with covariance of  $C_{\Z_i} =\sigma_i^2 I_{m_i \times m_i}$. 

Now,  we define the vector of coefficients of all the $k$ signals, $ \Xc$,  the vector of all the samples, $\Sconc$,   the vector of all the observations, $\Yc$,... , as follows:  
\begin{align} \Xc&=[\mathbf{X}_1^\fordagger, \mathbf{X}_2^\fordagger, \cdots, \mathbf{X}_k^\fordagger]^\fordagger,
\label{eqn:defnsA1}\\ \Sconc&=[\mathbf{S}_1^\fordagger, \mathbf{S}_2^\fordagger, \cdots, \mathbf{S}_k^\fordagger]^\fordagger,
\label{eqn:defnsA2}\\ \Yc&=[\mathbf{Y}_1^\fordagger, \mathbf{Y}_2^\fordagger, \cdots, \mathbf{Y}_k^\fordagger]^\fordagger,
\label{eqn:defnsA3}\\ \Zc&=[\mathbf{Z}_1^\fordagger, \mathbf{Z}_2^\fordagger, \cdots, \mathbf{Z}_k^\fordagger]^\fordagger,
\label{eqn:defnsA4}\\ \Delc &=[\mathbf{\Delta}_1^\fordagger, \mathbf{\Delta}_2^\fordagger, \cdots, \mathbf{\Delta}_k^\fordagger]^\fordagger,
\label{eqn:defnsA5}\\ \Qa&=[{Q}_1^\fordagger, {Q}_2^\fordagger, \cdots, {Q}_k^\fordagger]^\fordagger.
\label{eqn:defnsA6}
\end{align}
Furthermore, let
 \begin{align}
\Qb= \bigoplus_{i=1}^k Q_i\label{defQbA}
\end{align}
 be the direct sum of the individual matrices $Q_i$.
Then, we can write
\begin{align}\Yc = \Sconc+\Zc=\Qb \Xc+\Zc=\Qa \mathbf{X}+\Qb \Delc+\Zc.\end{align}
 One can verify that the covariance matrix of the noise is  
\begin{align}
C_{\Zc} =\bigoplus_{i=1}^k C_{\Z_i}=\bigoplus_{i=1}^k\sigma_i^2 I_{m_i \times m_i}.
\label{eqn:defCz18}
\end{align} Moreover, $C_{\mathbf{X}}=I_{2N\times 2N}$ and $C_{\Xc} = \Lambda \otimes I_{2N\times 2N}$, where $\Lambda$ was given in  \eqref{Rho}, and $N=N_2-N_1+1$.

\subsection{Reconstruction of  the remote signal $S(t)$ and its corrupted version $S_i(t)$}
 Here, we first state a lemma, which is  frequently used in the formulation and proofs of our problem. The Lemma provides   the  two alternative forms of the  LMMSE  estimator and the mean squared error. Next we  use this lemma to  formulate the  reconstruction of the remote signal $S(t)$ and its corrupted versions $S_i(t)$  in the subsequent subsections.

\begin{lemma}{\cite{Kay}\label{lmmse}
Suppose that $\Y = A\X+\Z$ in which $\Y$ is an observation vector, $A$ is a known matrix, $\X$ is a vector to be estimated and $\Z$ is an additive noise vector. 
In the case $\X$ and $\Z$ are mutually independent Gaussian vectors, LMMSE is optimal and the estimator and  the mean square error, respectively,  are given by
\begin{align}
\hat{\mathbf{x}}_{\mathrm{MMSE}}(\mathbf{y}) & = \mathbb{E} \left \{\mathbf{X} | \mathbf{y} \right \}=W\mathbf{y}, \nonumber \\
\mathbb{E}\| \mathbf{X}-\hat{\mathbf{X}} \|^2 & =\mathbb{E}_{\mathbf{Y}}\left \{ \mathsf{Var}[\mathbf{X} | \mathbf{Y}]\right\}=\tr(C_e), \label{eqn:mmseCe}
\end{align}
where the reconstruction matrix, $W$, and the error covariance matrice, $C_e$, are of the following forms:
\begin{align}
 W & = C_{\X\Y} C_{\Y}^{-1} = C_{\X}A^\fordagger(AC_{\X}A^\fordagger+C_{\Z})^{-1}, \nonumber
\\C_e &= C_{\X} -C_{\X\Y} C_{\Y}^{-1}C_{\Y\X}= C_{\X} -C_{\X}A^\fordagger (AC_{\X}A^\fordagger+C_{\Z})^{-1} AC_{\X}.\label{ce1}
\end{align}
Or alternatively \cite{Luenberger}, using the matrix identity
\begin{align}
C_{\X} A^\fordagger ( A C_{\X}A^\fordagger +C_{\Z})^{-1}  =  ( A^\fordagger C_{\Z}^{-1} A + C_{\X}^{-1})^{-1}A^\fordagger C_{\Z}^{-1},\label{alternativeformmatrix}
\end{align}
the matrices $W$ and $C_e$ are  given by
\begin{align}
 W & =  ( A^\fordagger C_{\Z}^{-1} A + C_{\X}^{-1})^{-1}A^\fordagger C_{\Z}^{-1},\nonumber
\\C_e &=  ( A^\fordagger C_{\Z}^{-1} A + C_{\X}^{-1})^{-1}.\label{ce2}
\end{align}}
\end{lemma}

\subsubsection{Reconstruction of $S(t)$} \label{secReconSt}
Here, the goal is to reconstruct $S(t)$ with minimum distortion using the observation vector $\Yc$. We use the MMSE criterion to minimize the average distortion subject to the samples. From the Parseval's theorem, we have
\begin{align}\Da=\frac{1}{T}\int_{t=0}^T\mathbb{E}\{|\hat{S}(t)-S(t)|^2\}dt= \frac 12 \mathbb{E}\|\mathbf{X}-\hat{\mathbf{X}}\|^2,\label{parseval1}\end{align}
where $\hat{S}(t)$ is the reconstructed signal and $\hat{\mathbf{X}}$ is the MMSE reconstruction of the coefficient vector $\mathbf{X}$ from the observation vector $\Yc$. Since  the random variables are jointly Gaussian, the MMSE estimator is optimal. 
From the equation 
$$\Yc=\Qa \mathbf{X}+\Ztilde,$$ where 
$\Ztilde=\Qb \Delc+\Zc$,
the error of the linear MMSE estiamtor is equal to 
\begin{align}
\mathbb{E}\| \mathbf{X}-\hat{\mathbf{X}}\|^2 & =\mathbb{E}_{\Yc}\left \{ \mathsf{Var}[\mathbf{X} | \Yc]\right\}=\tr(C^{\mathsf a}_e), \label{eqn:mmseCe1}
\end{align}
where $C^{\mathsf a}_e$ has the following two alternative forms
\begin{align}
C^{\mathsf a}_e &=  C_{\mathbf{X}}- C_{\mathbf{X}}   \Qa^\fordagger(\Qa  C_{\mathbf{X}} \Qa^\fordagger+C_{\Ztilde})^{-1} \Qa  C_{\mathbf{X}},   \label{Ce1St}
\\&= ( \Qa^\fordagger C_{\Ztilde}^{-1}  \Qa+C_{\mathbf{X}}^{-1}  )^{-1}. \label{Ce2St}
\end{align}
In the above formula, the covariance matrix of $\Ztilde$ is 
\begin{align}C_{\Ztilde}= \Qb C_{\Delc} \Qb^\fordagger +C_{\Zc}=\eta  \Qb  \Qb^\fordagger +C_{\Zc}.\label{defnCtildeZ}\end{align}

\subsubsection{Reconstruction of $S_i(t)$ for $i=1,2,\cdots, k$}
Here, the goal is to reconstruct all the $k$ signals with minimum distortion using the observation vector $\Yc$. Again, from the Parseval's theorem, we have 
\begin{align}\mathsf{D}_{\mathsf b}=\frac{1}{T}\sum_{i=1}^k\int_{t=0}^T\mathbb{E}\{|\hat{S}_i(t)-S_i(t)|^2\}{d}t= \frac 12 \mathbb{E}\|\Xc-\hat{\Xc}\|^2,\label{parseval2}
\end{align}
in which $\hat{S}(t)$ and  $\hat{\Xc}$ are the  reconstructed signal and  the estimated coefficients, respectively.
From the equation 
$$\Yc=\Qb \Xc+\Zc,$$ 
the LMMSE error is equal to 
\begin{align}
\mathbb{E}\| \Xc-\hat{\Xc}^2 \| & =\mathbb{E}_{\Yc}\left \{ \mathsf{Var}[\Xc | \Yc]\right\}=\tr(C^{\mathsf b}_e), \label{eqn:mmseCe}
\end{align}
where $C^{\mathsf b}_e$ has the following two alternative forms
\begin{align}
C^{\mathsf b}_e &=  C_{\Xc}- C_{\Xc}   \Qb^\fordagger(\Qb  C_{\Xc} \Qb^\fordagger+C_{\Zc})^{-1} \Qb  C_{\Xc},   \label{Ce1}
\\&= ( \Qb ^\fordagger C_{\Zc}^{-1}  \Qb+ C_{\Xc}^{-1}  )^{-1}. \label{Ce2}
\end{align}

\subsection{Some helpful facts}\label{section:helpful-facts}
We provide a number of facts about the matrices that we have introduced before. These facts can be directly verified, and will be repeatedly used in the proofs. We have listed these facts here to improve the presentation of the proofs. 
\begin{enumerate}[label=(\roman*)]
\item  \label{fact0} We have
\begin{align} \Qb  \Qb^\fordagger&=\bigoplus_{i=1}^k  Q_i  Q_i^\fordagger=\begin{pmatrix}{ {Q}_{1} {Q}_{1}^\fordagger} & 0 &0& \cdots & 0
\\0 & { {Q}_{2} {Q}_{2}^\fordagger }&0& \cdots&0
\\\vdots&\vdots&\vdots&\vdots&\vdots
\\0 & 0 &0& \cdots & { {Q}_{k} {Q}_{k}^\fordagger} 
 \end{pmatrix}_{m\times m}
\end{align}
and
\begin{align}
\Qa  \Qa^\fordagger&=\begin{pmatrix}{ {Q}_{1} {Q}_{1}^\fordagger} & { {Q}_{1} {Q}_{2}^\fordagger } &{ {Q}_{1} {Q}_{3}^\fordagger  } & \cdots & { {Q}_{1} {Q}_{k}^\fordagger }
\\{Q_{2} Q_{1}^\fordagger } & {Q_{2} Q_{2}^\fordagger }&{Q_{2} Q_{3}^\fordagger } & \cdots&{Q_{2} Q_{k}^\fordagger}
\\\vdots&\vdots&\vdots&\vdots&\vdots
\\{Q_{k} {Q}_{1}^\fordagger} & {Q_{k} Q_{2}^\fordagger} &{Q_{k} Q_{3}^\fordagger} & \cdots & {Q_{k} Q_{k}^\fordagger} 
 \end{pmatrix}_{m\times m},\end{align}
where $m=\sum_{i=1}^{k} m_i$.
\item \label{fact1} The rows of  matrix $Q_{i}$ are vectors of norm $\sqrt{N}$. Therefore, the matrix $Q_{i} Q_{i}^\fordagger$, for $i = 1, 2,\cdots, k$, is of size $m_i\times m_i$ and with diagonal entries equal to $N$ regardless of the value of $t_{ij}$. Therefore, $(\Qa  \Qa^\fordagger)_{\diag}=(\Qb  \Qb^\fordagger)_{\diag}=NI_{m\times m}$.

\item \label{fact2} When, $t_{ij}\in\{0,T/N, \cdots, (N-1)T/N \}$ and are distinct, the rows of  matrix $ {Q}_{i}$ will be perpendicular to each other. Therefore, the matrix $Q_{i} Q_{i}^\fordagger$ will be equal to $NI_{m_i\times m_i}$. Similarly, if  $t_{ij}$ are distinct for all $i,j$, the rows of $Q_{i}$ and $Q_{j}$ for $i\neq j$ will be perpendicular to each other. Therefore, $Q_{i} Q_{j}^\fordagger=0$ for $i\neq j$ in this case. Hence, using the definition of $\Qa$ and  $\Qb$
given in \eqref{eqn:defnsA6} and \eqref{defQbA}, both $\Qa \Qa^\fordagger$ and $\Qb \Qb^\fordagger$ will become diagonal matrices $NI_{m\times m}$, where $m=\sum_{i=1}^{k} m_i$. 

\item \label{fact3}  The diagonal elements of matrices $Q_i^\fordagger Q_i$ for $i = 1, 2, \cdots, k$ give us the norm of the column vectors of $Q_i$. They can be calculated as follows:
\begin{align}
Q_i^\fordagger Q_i(l,l) =\begin{cases}\sum_{j=1}^{{m_i}}\cos^2\left((N_1+l-1)\omega t_{ij}\right);&~~\text{for }1\leq l\leq N
\\
\\\sum_{j=1}^{m_i}\sin^2\left((N_1+l-N-1)\omega t_{ij}\right);&~~\text{for }N+1\leq l\leq 2N\end{cases}\label{QidaggerQi}.
\end{align}
When we use  the uniform sampling strategy, i.e.,  $\{t_{ij}\}=\{0, T/m_i, 2T/m_i, \cdots, (m_i-1)T/m_i\}$, the diagonal entries of $Q_i^\fordagger Q_i$ will become equal to $m_i/2$. This is because, for instance,
\begin{align}\sum_{j=1}^{{m_i}}\cos^2\left((N_1+l-1)\omega t_{ij}\right)&=\sum_{j=0}^{{m_i-1}}\cos^2\left((N_1+l-1)\omega \frac{T}{m_i}j\right)\nonumber
\\&=\sum_{j=0}^{{m_i-1}}\left(\frac{1}{2}+\frac{1}{2}\cos\Big(2(N_1+l-1)\frac{2\pi}{m_i}j\Big)\right)\nonumber
\\&=\frac{m_i}{2},\label{mi2rr}
\end{align}
where  \eqref{mi2rr} follows from the fact that $m_i > 2N_2$. Moreover, the off-diagonal entries will be zero.   For example, consider the entry
\begin{align}
Q_i^\fordagger Q_i(2,3) = \sum_{j=1}^{{m_i}} \cos\left((N_1+1)\omega t_{ij}\right)\cos\left((N_1+2)\omega t_{ij}\right)
&=\sum_{j=0}^{{m_i-1}}  \frac12 \left( \cos( \frac{\omega T}{m_i}j)+\cos\big((2N_1+3) \frac{\omega T}{m_i}j\big)\right)\nonumber
\\&=\sum_{j=0}^{{m_i-1}}  \frac12 \left( \cos( \frac{2\pi}{m_i}j)+\cos\big((2N_1+3) \frac{2\pi}{m_i}j\big)\right)\nonumber
\\&=0.\label{mi2}
\end{align}
In fact, with uniform sampling, different columns of the matrix $Q_i$ will be perpendicular to each other and $Q_i^\fordagger Q_i$ will become $(m_i/2)I_{2N\times 2N}$.
\end{enumerate}

\section{Proofs} \label{proofs} 

In this section, we state the proofs of our results. In the body of the proofs, we have used some lemmas, which are provided in  the appendix.

\subsection{Proof of Equation \eqref{Lower1S}} \label{proofsLower1S}

We start by computing the average distortion using equations \eqref{parseval1}, \eqref{eqn:mmseCe1} and  \eqref{Ce1St} as follows
\begin{align}2\mathsf{D}_{\mathsf a}=\tr (C^{\mathsf a}_e ) &= \tr \left( C_{\mathbf{X}}- C_{\mathbf{X}}   \Qa^\fordagger(\Qa  C_{\mathbf{X}} \Qa^\fordagger+C_{\Ztilde})^{-1} \Qa  C_{\mathbf{X}}\right)\nonumber\\
&=\tr (C_{\mathbf{X}})- \tr \left((\Qa  C_{\mathbf{X}} \Qa^\fordagger+C_{\Ztilde})^{-1} \Qa  C_{\mathbf{X}}^2  \Qa^\fordagger \right)\label{traceprop}\\
&= 2N - \tr \left((\Qa  \Qa^\fordagger+C_{\Ztilde})^{-1} \Qa   \Qa^\fordagger \right),\label{Cxidentity}
\end{align}
where  \eqref{traceprop} results from the cyclic property of the trace  and \eqref{Cxidentity} follows from the fact that $C_{\mathbf{X}}=I_{2N\times 2N}$.

We would like to show that for any arbitrary choice of sampling time instances, $t_{ij}$,  the average distortion will be  bounded from below as follows:
\begin{align}2\mathsf{D}_{\mathsf a} &= 2N- \tr \left[ (\Qa  \Qa^\fordagger+C_{\Ztilde})^{-1} \Qa  \Qa^\fordagger \right]  \nonumber
\\&\geq 2N - N \sum_{i=1}^{k} \frac{m_i}{(1+\eta)N+\sigma_i^2}.\label{used}
\end{align}
In other words, from \eqref{Cxidentity}, we wish to prove that
\begin{align} \tr\left((\Qa  \Qa^\fordagger+C_{\Ztilde})^{-1} \Qa   \Qa^\fordagger \right) \leq  \sum_{i=1}^{k} \frac{Nm_i}{(1+\eta)N+\sigma_i^2}.\label{QCxQinvQCxQ1}
\end{align}
Equivalently, if we use \eqref{defnCtildeZ} to replace $C_{\Ztilde}$ with $(\eta \Qb \Qb^\fordagger+C_{\Zc})$, we would like to show that
\begin{align} \tr\left((\Qa  \Qa^\fordagger+\eta \Qb \Qb^\fordagger+C_{\Zc})^{-1} \Qa   \Qa^\fordagger \right) \leq  \sum_{i=1}^{k} \frac{Nm_i}{(1+\eta)N+\sigma_i^2}.\label{QCxQinvQCxQ1N}
\end{align}
This can be derived using  Theorem  \ref{Lemmaab} (given in Appendix \ref{app}) with matrices $F=\Qa  \Qa^\fordagger+\eta \Qb \Qb^\fordagger$, $G=\Qa  \Qa^\fordagger$ and $C=C_{\Zc}$. From Fact \ref{fact0} of Section \ref{section:helpful-facts}, observe that  the matrices $F$ and $G$ are of the forms
\begin{align} 
F=\begin{pmatrix}{(1+\eta)Q_1 Q_1^\fordagger} & {Q_1 Q_2^\fordagger } &{Q_1 Q_3^\fordagger  } & \cdots & {Q_1 Q_k^\fordagger }
\\{Q_2 Q_1^\fordagger } & {(1+\eta)Q_2 Q_2^\fordagger }&{Q_2 Q_3^\fordagger } & \cdots& {Q_2 Q_k^\fordagger}
\\\vdots&\vdots&\vdots&\vdots&\vdots
\\{Q_k Q_1^\fordagger} & {Q_k Q_2^\fordagger} &{Q_k Q_3^\fordagger} & \cdots & {(1+\eta)Q_k Q_k^\fordagger} 
 \end{pmatrix}_{m\times m}\label{RhoNN-1}
\end{align}
and
\begin{align} 
G=\Qa  \Qa^\fordagger=\begin{pmatrix}{Q_1 Q_1^\fordagger} & {Q_1 Q_2^\fordagger } &{Q_1 Q_3^\fordagger  } & \cdots & {Q_1 Q_k^\fordagger }
\\{Q_2 Q_1^\fordagger } & {Q_2 Q_2^\fordagger }&{Q_2 Q_3^\fordagger } & \cdots&{Q_2 Q_k^\fordagger}
\\\vdots&\vdots&\vdots&\vdots&\vdots
\\{Q_k Q_1^\fordagger} & {Q_k Q_2^\fordagger} &{Q_k Q_3^\fordagger} & \cdots & {Q_k Q_k^\fordagger} 
 \end{pmatrix}_{m\times m}\label{RhoNN2-1}. 
\end{align}
 These matrices satisfy the required properties of Theorem \ref{Lemmaab}, \emph{i.e.,} the matrices $F$ and $G$ are positive semi-definite and $G = F \circ L$, where the matrix $L$ has the form of  \eqref{RhoNN2d} with parameters $a= 1+\eta$ and $b=1$. Therefore, we have the following inequality
\begin{align} \tr \left[ \big( F+ C\big)^{-1}G\right ]& \leq \tr \left[ \big( F_{\diag}+ C\big)^{-1}G_{\diag}\right ].
\end{align}

 Hence,  from Fact \ref{fact1} of Section \ref{section:helpful-facts} which states that $F_\diag=(\Qa  \Qa^\fordagger)_\diag+\eta (\Qb \Qb^\fordagger)_\diag=N(1+\eta)I_{m\times m}$ and $G_\diag=(\Qa  \Qa^\fordagger)_\diag=NI_{m\times m}$, the desired inequality in \eqref{QCxQinvQCxQ1N} concludes.

Furthermore, when $\sum_{i=1}^km_i\leq N$,  we would like to show that this lower bound is tight  if  we take distinct time instances, $t_{ij}$, from the set  $\{0,T/N, \cdots, (N-1)T/N \}$. Observe that this is possible since $\{0,T/N, \cdots, (N-1)T/N \}$ has $N$ elements.
 Fact \ref{fact2} from  Section \ref{section:helpful-facts} states that  both the matrices $\Qa \Qa^\fordagger$ and $\Qb \Qb^\fordagger$ will become diagonal matrices $NI_{m\times m}$, and thus $C_{\Ztilde}$ given in 
\eqref{defnCtildeZ} will be 
\begin{align}
C_{\Ztilde} &=\eta \Qb \Qb^\fordagger+C_{\Zc}=\eta N I_{m \times m} +C_{\Zc}. \nonumber
\end{align}
 Hence, the minimum  distortion will be 
\begin{align}2\mathsf{D}_{\mathsf{a} \min}&=  2N-  N \cdot\tr\left[((1+\eta)N I +C_{\Zc})^{-1} \right], \nonumber
\\&=2N- \sum_{i=1}^{k} \frac{Nm_i}{(1+\eta)N+\sigma_i^2},\label{eqn:Czsubstituted}
\end{align}
where \eqref{eqn:Czsubstituted} is derived using the definition of the diagonal matrix $C_{\Zc}$ given in \eqref{eqn:defCz18}.
 \hfill\ensuremath{\square}

\subsection{Proof of equation \eqref{Lower2S}:} \label{proofsLower2S}

To compute the minimum average distortion, from  $\Da = 1/2  \tr(C^{\mathsf a}_e)$, we use the alternative form of LMMSE (given in \eqref{Ce2St}), in which $C^{\mathsf a}_e$ is of the form
\begin{align} 
C^{\mathsf a}_e  =  ( \Qa^\fordagger C_{\Ztilde}^{-1}  \Qa+C_{\mathbf{X}}^{-1}  )^{-1}.\label{form2}
\end{align}
Hence, the average distortion will be
\begin{align}
2\Da =   \tr(C^{\mathsf a}_e)  =  \tr\left[\left(\sum_{i =1}^{k} {Q_i^\fordagger (\eta Q_i Q_i^\fordagger +\sigma_i^2 I)^{-1}  Q_i}+ I\right)^{-1}\right],\label{lefthandside521}
\end{align}
which results from the facts that $C_{\Ztilde}=\eta \Qb \Qb^\fordagger +C_{\Zc}$ (Section \ref{secReconSt}) and $C_{\mathbf{X}} =I_{2N \times 2N}$.

Using the matrix identity given in \eqref{alternativeformmatrix} for matrix $A$ to be  $\sqrt\eta Q_i$, we have 
\begin{align}
 Q_i^\fordagger \left(\eta Q_i Q_i^\fordagger +C_{\Z_i}\right)^{-1}  Q_i
=  \left( \eta Q_i^\fordagger C_{\Z_i}^{-1} Q_i+I\right) ^{-1} Q_i^\fordagger C_{\Z_i}^{-1} Q_i.\nonumber
\end{align}

Therefore, the matrix in  the left-hand side  of \eqref{lefthandside521} will be
\begin{align}
\sum_{i =1}^{k} Q_i^\fordagger \left(\eta Q_i Q_i^\fordagger +C_{\Z_i}\right)^{-1}  Q_i
&= \sum_{i =1}^{k} \left( \eta Q_i^\fordagger C_{\Z_i}^{-1} Q_i+I\right) ^{-1} Q_i^\fordagger C_{\Z_i}^{-1} Q_i\nonumber
\\&=  \frac{1}{\eta} \sum_{i =1}^{k} \left(\eta Q_i^\fordagger Q_i+\sigma_i^2 I \right)^{-1} \left (\eta Q_i^\fordagger Q_i+\sigma_i^2 I-\sigma_i^2 I \right)\label{Czi}
\\&= \frac{1}{\eta} \sum_{i =1}^{k} \left(I - \left( \eta \frac{ Q_i^\fordagger Q_i}{\sigma_i^2} +I\right)^{-1}\right)\nonumber
\\& = \frac{1}{\eta} \sum_{i =1}^{k} (I-B_i)\label{Bi}
\\& = \frac{1}{\eta} (k I -\sum_{i =1}^{k}B_i)\nonumber
\\& \leq  \frac{1}{\eta} \left(kI -k^2\Big(\sum_{i =1}^{k}B_i^{-1}\Big)^{-1}\right),\label{Generalized}
\end{align}
where  \eqref{Czi} results  from the fact that $C_{\Z_i}=\sigma^2_i I$, and  in \eqref{Bi}  matrix $B_i $  stands for $\left( \eta { Q_i^\fordagger Q_i}/{\sigma_i^2} +I\right)^{-1} $.  Moreover, \eqref{Generalized} is derived using Lemma \ref{Generalizedmean} for positive definite matrices $B_i$ with the equality if and only if $B_1 = B_2 = \cdots = B_k$. Notice  that  $0 < B_i \leq  I$. 

For any two symmetric positive definite matrices $A$ and $B$,  the relation $A \leq  B$ implies that $B^{-1} \leq  A^{-1}$. This is because the function  $f(t)=-t^{-1}$  is operator monotone \cite{Eric}. 
 Hence,  \eqref{Generalized} implies that 
\begin{align}
\left(I+\sum_{i =1}^{k} {Q_i^\fordagger (\eta Q_i Q_i^\fordagger +\sigma_i^2 I)^{-1}  Q_i}\right)^{-1} 
&\geq    \left[ I+\frac{1}{\eta} \left({k} I -k^2\Big(\sum_{i =1}^{k}B_i^{-1}\Big)^{-1}\right) \right]^{-1}\nonumber.
\end{align}
Let $A=\sum_{i =1}^{k}B_i^{-1}=\sum_{i =1}^{k}\left( \eta { Q_i^\fordagger Q_i}/{\sigma_i^2} +I\right) $. Then,   the  relation between the traces of the above matrices is  
\begin{align}
\tr\left(I+\sum_{i =1}^{k} {Q_i^\fordagger (\eta Q_i Q_i^\fordagger +\sigma_i^2 I)^{-1}  Q_i}\right)^{-1} \nonumber
&\geq \tr\left[ I+\frac{1}{\eta} \left({k} I -k^2A^{-1}\right) \right]^{-1}
\\&\geq  \tr\left[ I+\frac{1}{\eta} \left({k} I -k^2A_{\diag}^{-1}\right) \right]^{-1}\label{PeilersBi},
\end{align}
where \eqref{PeilersBi} results from Lemma \ref{PeirelsLemma} for the convex function $f(t) = (1+k/\eta -k^2t^{-1}/\eta )^{-1}$ when $t \geq k $ and the Hermitian matrix $A\geq kI$. 
 To find $A_{\diag}$, we need to calculate the diagonal entries of matrix $A$. They are
 \begin{align}
A(l,l) =\begin{cases}\sum_{i=1}^{{k}}\left( 1+\sum_{j=1}^{{m_i}}\frac{\eta}{\sigma_i^2}\cos^2\big((N_1+l-1)\omega t_{ij}\big)\right)\triangleq  a_{ \ell};&~~\text{for }1\leq l\leq N
\\
\\\sum_{i=1}^{{k}}\left(1+\sum_{j=1}^{m_i}\frac{\eta}{\sigma_i^2}\sin^2\big((N_1+l-N-1)\omega t_{ij}\big)\right)\triangleq  b_{ \ell};&~~\text{for }N+1\leq l\leq 2N\end{cases}\label{AidaggerAi},
\end{align}
since the diagonal elements of matrices $Q_i^\fordagger Q_i$ for $i = 1, 2, \cdots, k$  are of the following forms (Fact \ref{fact3} from Section \ref{section:helpful-facts})
\begin{align}
Q_i^\fordagger Q_i(l,l) =\begin{cases}\sum_{j=1}^{{m_i}}\cos^2\left((N_1+l-1)\omega t_{ij}\right);&~~\text{for }1\leq l\leq N
\\
\\\sum_{j=1}^{m_i}\sin^2\left((N_1+l-N-1)\omega t_{ij}\right);&~~\text{for }N+1\leq l\leq 2N\end{cases}\label{QidaggerQi}.
\end{align}
Substituting the diagonal entries of the matrix $A$  in \eqref{PeilersBi}, we obtain
\begin{align}
\tr\left(I+\sum_{i =1}^{k} {Q_i^\fordagger (\eta Q_i Q_i^\fordagger +\sigma_i^2 I)^{-1}  Q_i}\right)^{-1} \nonumber
&\geq  \tr\left[ I+\frac{1}{\eta} \left({k} I -k^2A_{\diag}^{-1}\right) \right]^{-1}\nonumber
\\& = \sum_{\ell =1}^{N} \frac{1}{(1+\frac{k}{\eta})-\frac{k^2}{\eta} a_{\ell}^{-1}}+\frac{1}{(1+\frac{k}{\eta})-\frac{k^2}{\eta } b_{\ell}^{-1}}\nonumber
\\&\geq \sum_{\ell =1}^{N} \frac{2}{(1+\frac{k}{\eta})-\frac{k^2}{\eta} (k+\eta \sum_{i=1}^{k} \frac{m_i}{2\sigma_i^2})^{-1} }\label{arithmatic}
\\&=  \frac{2N}{(1+\frac{k}{\eta})-\frac{k^2}{\eta} (k+\eta \sum_{i=1}^{k} \frac{m_i}{2\sigma_i^2})^{-1} },\label{m1minimizesD}
\end{align}
where \eqref{arithmatic} results from convexity of the function $f(t) = (1-a t^{-1})^{-1}$ for $t \geq a $. 

Now suppose that each $m_i > 2N_2$  for $i =1, 2, \cdots, k$. If we uniformly sample the signals, \emph{i.e.}, sample  $S_i(t)$ at time instances $\{0, T/m_i, 2T/m_i, \cdots, (m_i-1)T/m_i\}$, from Fact \ref{fact3} of Section \ref{section:helpful-facts} we conclude that the equality in the above equations holds, and thus the this lower bound is achieved.

 \hfill\ensuremath{\square}

\subsection{Proof of Equation \eqref{Lower1corrupted}:} \label{proofsLower1Si}
To compute the average distortion, here we use equations \eqref{parseval2}, \eqref{eqn:mmseCe}  and \eqref{Ce2}. Hence,
\begin{align}2\Db&=\tr \left(C_{\Xc}- C_{\Xc}   \Qb^\fordagger(\Qb  C_{\Xc} \Qb^\fordagger+C_{\Zc})^{-1} \Qb  C_{\Xc}\right)\nonumber
\\&=\tr( C_{\Xc})- \tr\left((\Qb  C_{\Xc} \Qb^\fordagger+C_{\Zc})^{-1} \Qb  C_{\Xc}^2\Qb^\fordagger\right)\label{permute}
\\&= k(2N)(1+\eta) - \tr\left((\Qb  C_{\Xc} \Qb^\fordagger+C_{\Zc})^{-1} \Qb  C_{\Xc}^2 \Qb^\fordagger\right),\label{QbCxQb}
\end{align}
where  \eqref{permute} and \eqref{QbCxQb} are achieved, respectively,  by the trace  cyclic property  and the fact that 
$ C_{\Xc}= \Lambda \otimes I_{2N\times2N}$ (The matrix $\Lambda$ has been defined in \eqref{Rho}).  

Here we are interested in reconstructing the signals $S_i(t)$. 
Following the similar steps from Subsection \ref{proofsLower1S}, we use Theorem \ref{Lemmaab} with the choice of $F=\Qb  C_{\Xc} \Qb^\fordagger$, $G=\Qb  C_{\Xc}^2 \Qb^\fordagger$ and $C=C_{\Zc}$. To demonstrate that these  matrices have the required properties of Theorem \ref{Lemmaab},
one can verify (by explicit evaulation) that $F$ has the same expression as in \eqref{RhoNN-1}, \emph{i.e.},
$F=\Qb  C_{\Xc} \Qb^\fordagger=
\Qb  (\Lambda \otimes I_{2N\times2N}) \Qb^\fordagger
$ is also equal to $
\Qa  \Qa^\fordagger+\eta \Qb \Qb^\fordagger.
$ Furthermore to compute $G$, 
 observe that $C_X ^2= \Lambda^2 \otimes I_{2N\times2N}$ in which  ${\Lambda}^2$ is a matrix of the following form
\begin{align}
\Lambda^2(i,j)=\begin{cases} (1+\eta)^2+(k-1) \triangleq \alpha  ~~ ; ~~i=j,\label{eqalpha}
\\
\\ 2(1+\eta)+(k-2)\triangleq  \beta ~~;~~i\neq j,\end{cases}
\end{align}
for $i , j = 1,  2, \cdots, k$. Then, one can verify that  
\begin{align} 
G=\begin{pmatrix}{\alpha Q_1 Q_1^\fordagger} & {\beta Q_1 Q_2 ^\fordagger } &{\beta Q_1 Q_3^\fordagger  } & \cdots & {\beta Q_1 Q_k^\fordagger }
\\{\beta Q_2 Q_1^\fordagger } & {\alpha Q_2 Q_2^\fordagger }&{\beta Q_2 Q_3^\fordagger } & \cdots&{\beta Q_2 Q_k^\fordagger}
\\\vdots&\vdots&\vdots&\vdots&\vdots
\\{\beta Q_k Q_1^\fordagger} & {\beta Q_k Q_2^\fordagger} &{\beta Q_k Q_3^\fordagger} & \cdots & {\alpha Q_k Q_k^\fordagger} 
 \end{pmatrix}_{m\times m}\label{RhoNN2}.
\end{align}
 Therefore, applying  Theorem \ref{Lemmaab} for the matrices $F$, $G$ and $C=C_{\Zc}$, we have 
\begin{align} \tr\left((\Qb  C_{\Xc} \Qb^\fordagger+C_{\Zc})^{-1} \Qb  C_{\Xc}^2 \Qb^\fordagger \right) 
&= \tr[(F+C_{\Zc})^{-1}G]\nonumber
\\&\leq \tr[(F_\diag+C_{\Zc})^{-1}G_\diag]\nonumber
\\&=  \tr \left[\big( (1+\eta)N I +C_{\Zc}\big)^{-1}\alpha N I \right]\label{diagFG}
\\& =  \sum_{i=1}^{k} \frac{\alpha N m_i}{(1+\eta)N+\sigma_i^2},\label{QCxQinvQCxQ}
\end{align}
where  \eqref{diagFG} results from the fact that the diagonal entries of the matrices $F$ and $G$ are $(1+\eta)N$ and $\alpha N$, respectively (Fact \ref{fact1} from Section \ref{section:helpful-facts} and  \eqref{RhoNN2}).
Consequently, the average distortion, for any arbitrary choice of sampling times, can be  bounded  as
\begin{align}
2\mathsf{D}_{\mathsf b} &= 2Nk(1+\eta)- \tr \left( (\Qb  C_{\Xc} \Qb^\fordagger+C_{\Zc})^{-1} \Qb  C_{\Xc}^2 \Qb^\fordagger \right)  \nonumber
\\&\geq 2Nk(1+\eta) -\sum_{i=1}^{k} \frac{\alpha N m_i}{(1+\eta)N+\sigma_i^2},\nonumber
\end{align}\\
in which  $\alpha$ is the one  defined in \eqref{eqalpha}. 

Furthermore,  we would like to show that the equality holds when $\sum_{i=1}^km_i\leq N$ and we take distinct time instances, $t_{ij}$, from the set  $\{0,T/N, \cdots, (N-1)T/N \}$. To show that 
 we compute the average distortion when $t_{ij}$ are distinct and belong to $\{0,T/N, \cdots, (N-1)T/N \}$ for all $i,j$. Using Fact \ref{fact2} from  Section \ref{section:helpful-facts}, the two matrices $F$ and $G$ will become diagonal matrices of the forms:
\begin{align}
F=\Qb  C_{\Xc} \Qb^\fordagger = (1+\eta) NI_{m\times m},
~~~G=\Qb  C_{\Xc}^2 \Qb^\fordagger = \alpha NI_{m\times m}
\end{align}
 Hence, the  average distortion will be 
\begin{align}2\mathsf{D}_{\mathsf b\min} &= 2Nk(1+\eta)- \tr \left( ((1+\eta) NI_{m\times m}+C_{\Zc})^{-1} \alpha NI_{m\times m} \right) 
 \\&= 2Nk(1+\eta)- \sum_{i=1}^{k} \frac{\alpha N m_i}{(1+\eta)N+\sigma_i^2}\label{eqn:Czsub}
\end{align}
where \eqref{eqn:Czsub} is derived using the diagonal matrix $C_{\Zc}$ of \eqref{eqn:defCz18}.
\hfill\ensuremath{\square}

\subsection{ Proof of Equation \eqref{Lower2corrupted}:} \label{proofsLower2Si}
Here, we divide the proof into two parts. In the first part, we show  that 
$$\mathsf{D}_{ \mathsf b\min} \geq N \cdot\tr\left( \diag ([\frac{m_1}{2\sigma_1^2}, \frac{m_2}{2\sigma_2^2}, \cdots, \frac{m_k}{2\sigma_k^2}])+\Gamma\right)^{-1},$$
and then we show that when $m_i > 2N_2$ (for  each $i \in {1, 2, \cdots, k}$), 
the  optimal sampling strategy is uniform sampling and  the minimum distortion is equal to 
$$\mathsf{D}_{ \mathsf b\min}=N \cdot\tr\left( \diag ([\frac{m_1}{2\sigma_1^2}, \frac{m_2}{2\sigma_2^2}, \cdots, \frac{m_k}{2\sigma_k^2}])+\Gamma\right)^{-1}.$$
In the second part, we simplify the above equation to obtain the expression given in the statement of the theorem. 

\noindent
\textbf{Part (i):} 
To compute the minimum average distortion, we use $\Db = 1/2  \tr(C^{\mathsf b}_e)$ and   the alternative form of LMMSE,  
where $C^{\mathsf b}_e$ is of the form
\begin{align} 
C^{\mathsf b}_e  = \left(\Qb^\fordagger  C_{\Zc}^{-1} \Qb+C_{\Xc}^{-1}\right)^{-1}.\label{form2}
\end{align}
In the above formula,
\begin{align}  C_{\Zc}^{-1} =\bigoplus_{i=1}^k \frac{1}{\sigma^2_i}I_{2N\times2N}, \quad C_{\Xc}^{-1} = \Gamma_{k\times k} \otimes I_{2N\times2N},\end{align}
where the matrix $\Gamma$ is the inverse of the matrix $\Lambda$, given in \eqref{Rho}. Therefore, the average distortion will be
\begin{align}
2\Db =   \tr(C^{\mathsf b}_e)  =  \tr\left[\left(\bigoplus_{i =1}^{k}\frac{Q_i^\fordagger Q_i}{\sigma_i^2}\right)+\Gamma\otimes I\right]^{-1}.\label{eqn:L84}
\end{align}

First notice that due to Lemma \ref{Hornlemma}, matrices   $\Gamma\otimes I$ and  $I\otimes \Gamma$ are permutation similar, i.e., there exists a unique  permutation matrix $P(k, 2N)$ of size $2Nk \times 2Nk$ such that
$$\Gamma\otimes I =P(k, 2N)^T ( I \otimes \Gamma) P(k, 2N),$$ and moreover, $P(k, 2N)$ has the property 
$$P(k, 2N) = P(k, 2N)^T =P(k, 2N)^{-1}.$$ 
Using the permutation matrix $P(k, 2N)$ and  the cyclic property of  the trace, we have 
\begin{align}
2\Db =   \tr(C^{\mathsf b}_e)  &= \tr\left(P(k, 2N)^T C^{\mathsf b}_eP(k, 2N)\right)\nonumber
 \\ &=\tr\left[P(k, 2N)^T  \left(\bigoplus_{i =1}^{k}\frac{ Q_i^\fordagger Q_i}{\sigma_i^2}\right)P(k, 2N)+P(k, 2N)^T \left(\Gamma\otimes I\right) P(k, 2N)\right]^{-1}\nonumber
\\ &=\tr\left(H+I\otimes \Gamma\right)^{-1},
\end{align}
where matrix $H$ denotes the permuted matrix $P(k, 2N)^T (\bigoplus_{i =1}^{k}{ Q_i^\fordagger Q_i}/{\sigma_i^2})P(k, 2N)$. Notice that the matrix 
$$\bigoplus_{i =1}^{k}{ Q_i^\fordagger Q_i}=\sum_{i=1}^k G_i\otimes Q_i^\fordagger Q_i,$$
where $G_i$ is a $k\times k$ matrix defined as follows: $G_i(i',j')=0$ if $(i',j')\neq (i,i)$, and $G_i(i',j')=1$  if $(i',j')=(i,i)$. Therefore, the matrix $H$ can be written as 
\begin{align}H&=P(k, 2N)^T (\bigoplus_{i =1}^{k}{ Q_i^\fordagger Q_i}/{\sigma_i^2})P(k, 2N)\nonumber
\\&=\sum_{i=1}^k Q_i^\fordagger Q_i\otimes G_i.
\end{align}
If we partition $H$ into $k\times k$ submatrices $H_{ij}$ as follows
\begin{align}
H=\begin{pmatrix}
{H_{11}} &{H_{12}}&\cdots&H_{1(2N)}
\\{H_{21}} & {H_{22}}& \cdots&{H_{2(2N)}}
\\\vdots&\vdots&\vdots&\vdots
\\{H_{(2N)1}}&{H_{(2N)2}}& \cdots &{H_{(2N)(2N)}} \end{pmatrix},\label{MatrixH}
\end{align}
 all the $H_{ij}$ submatrices will be diagonal matrices because they are weighted sums of diagonal matrices $G_i$. More precisely, the submatrices $H_{ll}$ for $l\leq N$ can be computed as follows, using Fact \ref{fact3} from Section \ref{section:helpful-facts} that gives us the diagonal entries of $Q_i^\fordagger Q_i$ (the  entries of matrix $H_{ll}$ for $1 \leq l \leq N$  and   $N\leq l \leq 2N$ are the $l$th  and the $(l+N)$th diagonal entries of  matrices $Q_i^\fordagger Q_i/\sigma_i^2$, given in \eqref{QidaggerQi}), respectively):
\begin{align}
H_{ll}=\begin{pmatrix}
{\frac{1}{\sigma_1^2} \sum_{j=1}^{{m_1}}{\cos^2\scriptstyle \left((N_1+l-1)\omega t_{1j}\right)}} & {0 } &{0 } & \cdots & {0 }
\\{0 } & {\frac{1}{\sigma_2^2} \sum_{j=1}^{m_2}\cos^2\scriptstyle \left((N_1+l-1)\omega t_{2j}\right)}&{0 } & \cdots&{0 }
\\\vdots&\vdots&\vdots&\vdots&\vdots
\\{0 } &{0 }&{0 } & \cdots & {\frac{1}{\sigma_k^2} \sum_{j=1}^{{m_k}}\cos^2\scriptstyle \left((N_1+l-1)\omega t_{kj}\right)} \end{pmatrix}_{k\times k}\nonumber
\end{align}
and for $l=N+1 , N+2, \cdots, 2N$,
\begin{align}
H_{ll}=\begin{pmatrix}
{\frac{1}{\sigma_1^2} \sum_{j=1}^{{m_1}}\sin^2\scriptstyle \left((N_1+l-N-1)\omega t_{1j}\right)} &{0 }  &{0 }  & \cdots & {0 } 
\\{0 }  & {\frac{1}{\sigma_2^2} \sum_{j=1}^{{m_2}}\sin^2\scriptstyle \left((N_1+l-N-1)\omega t_{2j}\right)}&{0 } & \cdots&{0 } 
\\\vdots&\vdots&\vdots&\vdots&\vdots
\\{0 } & {0 }  &{0 }  & \cdots & {\frac{1}{\sigma_k^2}\sum_{j=1}^{{m_k}}\sin^2\scriptstyle \left((N_1+l-N-1)\omega t_{kj}\right)} \end{pmatrix}_{k\times k}\nonumber
\end{align}
Applying  Lemma \ref{PeirelsLemma} for the  Hermitian matrix $A$  and the convex function $f(x)= x^{-1}$ for $x > 0$,  we attain a lower bound on the average distortion as:
\begin{align}
2\Db =   \tr(C^{\mathsf b}_e) &=\tr\left(H+I\otimes \Gamma\right)^{-1},\nonumber
\\ &\geq  \tr\left(H_{\mathsf{B}\diag}+I\otimes \Gamma\right)^{-1},\label{eqnL91}
\end{align}
in which  the matrix $H_{\mathsf{B}\diag}$ is the block diagonal form of matrix $H$, where all the submatrices other than the $2N$  block diagonal submatrices, ${H}_{ll}$, are zero. Consequently, the matrix  $(H_{\mathsf{B}\diag}+I\otimes \Gamma)$ is of the form
 \begin{align}
(H_{\mathsf{B}\diag}+I\otimes \Gamma)=\begin{pmatrix}
{H_{11}+\Gamma} & {0 } &{ \cdots}  & {0 } 
\\{0 } & {H_{22}+\Gamma} & {\cdots} & {0 }  
\\ \vdots&\vdots&\vdots&\vdots
\\{0 }& {0 } & \cdots&{H_{(2N)(2N)}+\Gamma}\end{pmatrix}
_{2Nk\times 2Nk}.
\end{align}
Therefore,
$$\tr\left(H_{\mathsf{B}\diag}+I\otimes \Gamma\right)^{-1}=\sum_{l=1}^{2N}
\tr\left(H_{ll}+\Gamma\right)^{-1}.$$
Using  Lemma \ref{jointlyconvex} (given in Appendix \ref{app}), for any $1\leq l\leq N$, we obtain
\begin{align}\frac{1}{2}(\tr\left(H_{ll}+\Gamma\right)^{-1}+
\tr\left(H_{(l+N)(l+N)}+\Gamma\right)^{-1})
&\geq \tr\left(\frac{1}{2}(H_{ll}+H_{(l+N)(l+N)})+\Gamma\right)^{-1}
\\&=\tr\left(\diag([\frac{m_1}{2\sigma_1^2}, \frac{m_2}{2\sigma_2^2}, \cdots, \frac{m_k}{2\sigma_k^2}])+\Gamma\right)^{-1}.
\end{align}
Consequently, 
 \begin{align}\
\tr\left(H_{\mathsf{B}\diag}+I\otimes \Gamma\right)^{-1}=\sum_{l=1}^{2N}
\tr\left(H_{ll}+\Gamma\right)^{-1} \geq
2N \cdot\tr\left( \diag ([\frac{m_1}{2\sigma_1^2}, \frac{m_2}{2\sigma_2^2}, \cdots, \frac{m_k}{2\sigma_k^2}])+\Gamma\right)^{-1}.
\end{align}
Therefore, from \eqref{eqnL91} for any choice of sampling time instances $t_{ij}$, we obtain
\begin{align}
\Db\geq  N \cdot\tr\left( \diag ([\frac{m_1}{2\sigma_1^2}, \frac{m_2}{2\sigma_2^2}, \cdots, \frac{m_k}{2\sigma_k^2}])+\Gamma\right)^{-1}. \label{Dblowerupper}
\end{align}

Moreover, suppose that each $m_i > 2N_2$  for $i =1, 2, \cdots, k$. If we employ uniform strategy, the matrices $Q_i^\fordagger Q_i$ ($i = 1, 2, \cdots, k$)  become diagonal matrices with diagonal entries equal to $m_i/2$ (Fact \ref{fact3} from Section \ref{section:helpful-facts}). Then,  the matrix in  the right hand side of  \eqref{eqn:L84} will be
\begin{align}\left[\left(\bigoplus_{i =1}^{k}\frac{m_i}{2\sigma_i^2}I_{2N\times 2N}\right)+\Gamma\otimes I\right]^{-1}
&=\left(\diag ([\frac{m_1}{2\sigma_1^2}, \frac{m_2}{2\sigma_2^2}, \cdots, \frac{m_k}{2\sigma_k^2}])\otimes I+\Gamma\otimes I\right)^{-1}
\\&=\left(\diag ([\frac{m_1}{2\sigma_1^2}, \frac{m_2}{2\sigma_2^2}, \cdots, \frac{m_k}{2\sigma_k^2}])+\Gamma\right)^{-1}\otimes I_{2N\times 2N}.
\end{align}
Therefore, from \eqref{eqn:L84} and the fact that $\tr (A \otimes B)= \tr(A) \tr(B)$, we get
$$\mathsf{D}_{\mathsf b\min}\leq N \cdot\tr\left( \diag ([\frac{m_1}{2\sigma_1^2}, \frac{m_2}{2\sigma_2^2}, \cdots, \frac{m_k}{2\sigma_k^2}])+\Gamma\right)^{-1}.$$

Therefore, from \eqref{} and the above inequality, we conclude that 
$$\mathsf{D}_{\mathsf b\min} =  N \cdot\tr\left( \diag ([\frac{m_1}{2\sigma_1^2}, \frac{m_2}{2\sigma_2^2}, \cdots, \frac{m_k}{2\sigma_k^2}])+\Gamma\right)^{-1}.$$

\noindent
\textbf{Part (ii):} \label{simplified} 
So far we have shown that  when $m_i > 2N_2$ for $i =1, 2, \cdots, k$, the minimal distortion is 
$$\mathsf{D}_{\mathsf b\min}=N \cdot\tr\left( \diag ([\frac{m_1}{2\sigma_1^2}, \frac{m_2}{2\sigma_2^2}, \cdots, \frac{m_k}{2\sigma_k^2}])+\Gamma\right)^{-1}.$$

 Here, we  wish to simplify the above equation to obtain 
Observe that 
 $\Gamma=\Lambda^{-1}$ can be computed from \eqref{Rho} as follows:
\begin{align}
\Gamma=\begin{pmatrix}{\aaa} & {\bbb } &{\bbb} & \cdots & {\bbb}
\\{\bbb} & {\aaa}&{\bbb} & \cdots&{\bbb}
\\\vdots&\vdots&\vdots&\vdots&\vdots
\\{\bbb} & {\bbb} &{\bbb} & \cdots & {\aaa}  \end{pmatrix}_{k\times k},\label{eqn:Rho-1}
\end{align}
in which
\begin{align}
 \aaa = \frac{\eta+k-1}{\eta(\eta+k)}~~ \text{ and } ~~\bbb = \frac{-1}{\eta(\eta+k)}.
\end{align}
For simplicity define:
\begin{align}
A := \diag([\frac{m_1}{2\sigma_1^2}, \frac{m_2}{2\sigma_2^2}, ..., \frac{m_k}{2\sigma_k^2}])+ \Gamma
\end{align}
Therefore, we need to find $\tr(A^{-1})$ which is equal to $\sum_{i=1}^{k} (A^{-1})_{i,i}$. To calculate the diagonal entries of $A^{-1}$, we use the Cramer's rule as follows:
\begin{align}
(A^{-1})_{i,i} = \frac{\det(A_i)}{\det(A)} 
\end{align}  
where $A_i$ is the remaining matrix after removing the $i$th row and column of $A$. According to Lemma \ref{det_lemma} (given in Appendix \ref{app}), we deduce:
\begin{align} \label{tr_inverse_i}
(A^{-1})_{i,i} = \frac{\prod\limits_{\substack{j=1 \\ j \neq i}}^{k}\phi_j+\bbb\sum\limits_{\substack{j=1 \\ j \neq i}}^{k}\prod\limits_{\substack{l=1 \\ l \neq i, j}}^{k} \phi_l} 
{\prod\limits_{j=1}^{k}\phi_j+\bbb\sum\limits_{j=1}^{k}\prod\limits_{\substack{l=1 \\ l \neq j}}^{k} \phi_l}
\end{align} 
in which $
\phi_i=\frac{m_i}{2\sigma_i^2}+\aaa-\bbb=\frac{m_i}{2\sigma_i^2}+\frac{1}{\eta}
$. Hence, 
\begin{align}\tr(A^{-1}) &= \sum_{i=1}^{k} (A^{-1})_{i,i}\label{Aii}
\\&= \frac{\sum\limits_{i=1}^{k}\frac{1}{\phi_i} +\bbb \sum\limits_{\substack{i,j=1 \\ j \neq i}}^{k} \frac{1}{\phi_i\phi_j}}{1+\bbb(\sum\limits_{i=1}^{k}\frac{1}{\phi_i})},\label{sumAii}
\end{align}
in which \eqref{sumAii} is derived by  replacing \eqref{tr_inverse_i} in \eqref{Aii} and dividing both numerator and denominator by $\prod\limits_{i=1}^{k}\phi_i$. 
Observe that by our definition $\Phi_{-1}: = \sum\limits_{i=1}^{k}\frac{1}{\phi_i}$, we get
\begin{align}
\tr(A^{-1})= \frac{\Phi_{-1}+\bbb(\Phi_{-1}^2-\sum\limits_{i=1}^{k}\frac{1}{\phi_i^2})}{1+\bbb \Phi_{-1}} = \Phi_{-1} - \frac{\bbb}{1+\bbb \Phi_{-1}}\Phi_{-2}.
\end{align}
We get the desired result by replacing $\bbb = -1/(\eta(\eta+k))$. \hfill\ensuremath{\square}

\appendix 

\section{Lemmas}\label{sec:prelim}
In this section, we state some  lemmas  that have been  used in the proof section.

\subsection{Majorization inequalities}
A vector $\x\in \mathbb{R}^n$ is majorized by $\y\in \mathbb{R}^n$ if after sorting the two vectors in decreasing order, the following inequalities hold:
\begin{align}
\sum_{i=1}^{k} x_i \leq \sum_{i=1}^{k} y_i  \quad (1 \leq k \leq n), \quad
\sum_{i=1}^{n} x_i = \sum_{i=1}^{n} y_i.  
\end{align}
A fundamental result in majorization theory states that for any Hermitian matrix $A$ of size $n \times n$, the diagonal entries of $A$ are majorized by its eigenvalues \cite{Bhatia}.
The extension of the above result to the block Hermitian matrices is also true (e.g. see \cite[Sec. 1]{Lin}):

\begin{lemma}[Block majorization inequality]\label{LemmaBlockMaj}
\emph{If a Hermitian matrix $A$ is partitioned into block matrices 
\begin{align}A=\begin{pmatrix}M_{11}&M_{12}&\dots &M_{1k}\\M_{21}&M_{22}&\dots&M_{2k}\\ \vdots &&\vdots\\ M_{k1}&M_{k2}&\dots &M_{kk}\end{pmatrix}\label{blockAM}\end{align}
for matrices $M_{ij},  {i,j=1,2,\dots, k}$, then the eigenvalues of $\bigoplus_{i=1}^kM_{ii}$ are majorized by the eigenvalues of $A$.}
\end{lemma}

If a vector $\x$ is majorized by $\y$, then for any convex functions $f:\mathbb R\mapsto\mathbb R$, we have  $\sum_{i}f(x_i) \leq \sum_if(y_i)$~\cite{Bhatia}. This implies that
\begin{lemma}\label{PeirelsLemma}
\emph{Let $\Omega$ be a closed interval in $\mathbb{R}$. 
For any Hermitian matrix $A$ with eigenvalues in $\Omega$, and any convex function $f$ on $\Omega$,
\begin{align}
\tr(f(A)) \geq \tr(f(A_{\diag})).\label{Peirels}
\end{align}
 More generally by Lemma \ref{LemmaBlockMaj} , for a  Hermitian matrix $A$ partitioned into  block matrices $M_{ij},  {i,j=1,2,\dots, k}$, as in \eqref{blockAM},
\begin{align}
\tr(f(A)) \geq \tr\left(f\Big(\bigoplus_{i=1}^kM_{ii}\Big)\right).\label{BlockPeirels}
\end{align} }
\end{lemma}

\subsection{Other useful definitions and inequalities}

\begin{lemma}\emph{\cite{Sagae}\label{Generalizedmean} 
Let $w_1, w_2, \cdots, w_k$ be non-negative weights adding up to one, and let $B_1, B_2, \cdots, B_k$ be $n\times n$ positive definite matrices. Consider the weighted arithmetic and harmonic means of the matrices $B_i$
\begin{align}
A &\triangleq w_1 B_1+w_2 B_2+ \cdots+ w_k B_k,\\
H &\triangleq (w_1 B_1^{-1}+w_2 B_2^{-1}+ \cdots+ w_k B_k^{-1})^{-1}.
\end{align}
Then, the following inequality holds,
$$H \leq  A,$$  
with  equality if and only  if  $B_1=B_2=\cdots=B_k$.}
\end{lemma}

\begin{definition}\emph{\cite{Zhan}
A real-valued continuous function $f(t)$ on a real interval $I$ is called operator monotone if
\begin{align}
A \leq  B	~~ \Rightarrow	~~f(A) \leq  f(B),
\end{align}
for Hermitian matrices $A$ and $B$ with eigenvalues in $I$ . 
Furthermore, $f$ is called operator convex if 
$$f(\lambda A + (1 - \lambda)B) \leq  \lambda f(A) + (1 - \lambda)f(B),$$
for any  $0 \leq \lambda \leq 1$ and Hermitian matrices $A$ and $B$ with  eigenvalues that are contained in $I$, and  $f$ is said to be operator concave if $-f$ is operator convex.}
\end{definition}

\begin{lemma} \emph{ \cite[p.260]{Horn} \label{Hornlemma}
Let $m$ and  $n$ be given positive integers and matrices $A$ and $B$ be any square matrices of sizes $m \times m$ and  $n \times n$, respectively. Then, matrix $B \otimes A$ is permutation similar to matrix $A\otimes B$, i.e.,  there is a unique matrix $P$ such that
\begin{align}
B \otimes A = P(m,n)^T (A\otimes B )P(m,n)
\end{align}
 where $P(m,n)$ is the following $mn\times mn$ permutation matrix:
\begin{align} 
 P(m,n) = \sum_{i=1}^{n}\sum_{j=1}^{m} E_{ij}\otimes E_{ij}^T                               \end{align}
in which  $E_{ij}$ is an $m \times n$ matrix such that only the $(i,j)$th entry  is  unity and the other entries are zero.
 Furthermore, the useful following property holds
\begin{align}
P(m,n) =P(n,m)^T = P(n,m)^{-1}.\label{propertypermut}
\end{align}}
\end{lemma}

\begin{lemma}\emph{\label{jointlyconvex} 
 Assume that $B_1, B_2$ are positive definite matrices. Let $B=(B_1+B_2)/2$, then $$2\tr(B^{-1})\leq \tr( B_1^{-1})+\tr(B_2^{-1}).$$ }
\end{lemma}

\begin{proof}
The  L\"{o}wner-Heinz theorem  implies that the function $f(t)=t^{-1}$ for $t>0$ is operator convex \cite{Eric}. From the fact that   $B=(B_1+B_2)/2$, we conclude the desired inequality.
\end{proof}

\begin{lemma} \label{det_lemma}\emph{
Given real non-negative $a_1,a_2,...,a_k$ and positive $b$, let
 \begin{align}
M(a_1,..,a_k,b) := \begin{pmatrix}{a_1} & {-b } &{-b } & \cdots & {-b }
\\{-b} & {a_2}&{-b} & \cdots&{-b}
\\\vdots&\vdots&\vdots&\vdots&\vdots
\\{-b} & {-b} &{-b} & \cdots & {a_k}  \end{pmatrix}_{k\times k}.
\end{align}
Then,
\begin{align}
\det(M(a_1,...,a_k,b)) = \prod_{i=1}^{k} (a_i+b) - b \sum_{i=1}^{k} \prod_{\substack{j=1 \\ j \neq i}}^{k} (a_j+b).
\end{align}}
\end{lemma}
\begin{proof} The elementary row operations do not change the determinant. If we first 
subtract the first row from all the other rows, and then multiply  the $i$th row of the matrix  by ${b}/{(a_i+b)}$  for $i=2,3,..,k$,  and add it to the first row, we end up with an upper triangular matrix with diagonal entries $\{a', a_2+b, a_3+b, ..., a_k+b\}$, where 
$
a' = a_1 - \sum_{i=2}^{k} \frac{(a_1+b)b}{(a_i+b)}.
$ 
Since the determinant of an upper triangular matrix is equal to product of the diagonal elements, we have
\begin{equation}
\begin{split}
\det(M(a_1,...,a_k,b)) &= a'(a_2+b)...(a_k+b)\\
 &=a_1(a_2+b)...(a_k+b) - b \sum_{i=2}^{k} \prod_{\substack{j=1 \\ j \neq i}}^{k} (a_j+b)  \\ 
 & =\prod_{i=1}^{k} (a_i+b) - b \sum_{i=1}^{k} \prod_{\substack{j=1 \\ j \neq i}}^{k} (a_j+b).
 \end{split}
 \end{equation}
 \end{proof}

\section{ A new reverse majorization inequality} \label{app}
\begin{theorem} \label{Lemmaab}
\emph{Take two positive semidefinite matrices $F$ and $G$ of sizes $m \times m$ satisfying $G=F\circ L$, where $\circ$ is the Hadamard product and $L$ is a matrix of the following form:
\begin{align} 
L=\begin{pmatrix}{a \mathbf{1}_{m_1\times m_1}}& {b\mathbf{1}_{m_1\times m_2} } &\cdots & {b\mathbf{1}_{m_1\times m_k} }
\\{b\mathbf{1}_{m_2\times m_1} } & {a\mathbf{1}_{m_2\times m_2} }& \cdots&{b\mathbf{1}_{m_2\times m_k}}
\\\vdots&\vdots&\vdots&\vdots
\\{b \mathbf{1}_{m_k\times m_1}} & {b \mathbf{1}_{m_k\times m_2}}  & \cdots & {a\mathbf{1}_{m_k\times m_k}} 
 \end{pmatrix},\label{RhoNN2d}
\end{align}
where $\mathbf{1}$ is a matrix with all one coordinates and $a$ and $b$ are two positive real numbers, where $0 \leq a \leq b$.
Then, for any positive definite diagonal matrix $C$, we have
\begin{align} \tr \left[ \big( F+ C\big)^{-1}G\right ]& \leq \tr \left[ \big( F_{\diag}+ C\big)^{-1}G_{\diag}\right ],
\label{diagonalab}
\end{align}
where $F_{\diag}$ is a diagonal matrix formed by taking the diagonal entries of $F$, and the matrix $G_{\diag}$ is defined similarly.}
\end{theorem}

\begin{proof}
If the statement of theorem holds for the matrix $F$, it will also hold for the matrix $kF$ for any positive constant $k$. Therefore, without loss of generality, we assume that $a=1$ and hence, $b \geq 1$.  From the Hadamard product relation, this implies that $F$ and $G$ are equal on block matrices on the diagonal. Let $A$ denote this common part, \emph{i.e.,}
$$A=\begin{pmatrix}{F_{1:m_1\times 1:m_1}}& 0 &0& \cdots & 0
\\0& {F_{(m_1+1:m_1+m_2)\times (m_1+1:m_1+m_2)} }&0& \cdots&0
\\\vdots&\vdots&\vdots&\vdots&\vdots
\\0 & 0 &0& \cdots & {F_{(m_1+\cdots+m_{k-1}+1)\times (m_1+\cdots+m_{k-1}+m_k)}} 
 \end{pmatrix}_{m\times m}.\label{RhoNN2d22}$$

One can find matrix $B$ such that 
$F=A+B$ and $G=A+b B$. Observe that wherever $A$ is non-zero, $B$ is zero and vice versa.
Substituting $F$ and $G$ in the left hand side of \eqref{diagonalab}, one attains
\begin{align}
 \tr \left[ \big( F+ C\big)^{-1}G\right ] &= \tr \left[ \big( A+ B+ C\big)^{-1} \big (A+b B \big )\right] \nonumber
\\&=\tr \left[(A+ B+ C)^{-1} (A+ B)\right ] +\tr \left[(A+ B+ C)^{-1}  (b-1)B )\right ]. \label{twoterms}
\end{align}
We will show that the first term in the above formula is less than or equal to the right hand side of \eqref {diagonalab}  and the second term is non-positive.

Start with the first term of \eqref{twoterms}
\begin{align}
 \tr \left[(A+ B+ C)^{-1} (A+ B)\right ] 
= \tr(I_{m\times m})- \tr \left((A+ B+ C)^{-1} C\right ),\label{abcab}
\end{align}
where the second term in \eqref{abcab} can be bounded as follows:
\begin{align}
\tr \left((A+ B+ C)^{-1} C\right )
 &= \tr \left [\big( C^{\frac{-1}{2}} A C^{\frac{-1}{2}} + C^{\frac{-1}{2}} BC^{\frac{-1}{2}} + I\big)^{-1} \right ]\nonumber
\\ &\geq \tr \left[\big( C^{\frac{-1}{2}} A_{\diag} C^{\frac{-1}{2}}  +I\big)^{-1} \right].\label{piereluse1}
\end{align}
The last inequality comes from Lemma  \ref{PeirelsLemma}. \\
Hence, \eqref{abcab} will be bounded as 
\begin{align}
 \tr \left[(A+ B+ C)^{-1} (A+ B)\right ] 
&\leq \tr(I_{m\times m})- \tr \left[\big( C^{\frac{-1}{2}} A_{\diag} C^{\frac{-1}{2}}  +I\big)^{-1} \right],\label{term1}
\\&= \tr[(A_{\diag} +C)^{-1}(A_{\diag}+C)] - \tr[C^{\frac{1}{2}} (A_{\diag}+C)^{-1} C^{\frac{1}{2}}]
\\&= \tr [(A_{\diag} +C)^{-1}A_{\diag}] \label{term1_b}
\\& = \tr [(F_{\diag} +C)^{-1}G_{\diag}], \label{desired}
\end{align}
in which \eqref{term1_b} comes from the trace interchange property and \eqref{desired} is derived since $A$ is defined to be the common part of the two matrices $F$ and $G$.
 
To complete the proof, it remains to  show that the right hand side of \eqref{twoterms} is non-positive, i.e.,
\begin{align} 
 ( b- 1 ) \cdot\tr \left[ (A+ B+ C)^{-1} B \right] \leq 0.\label{secondterm1}
\end{align}
Since we have assumed that $b\geq 1$, we  need to  show that  the trace function is non-positive.  For the positive definite matrix $(A+C)$, we have
\begin{align} 
\tr \left((A+B+C)^{-1} B \right) \nonumber
& =\tr(I) - \tr \left( (A+C+ B)^{-1} (A+C) \right)\nonumber
\\& =\tr(I) - \tr \left((A+C)^{\frac{1}{2}}(A+ C+B)^{-1} (A+C)^{\frac{1}{2}} \right)\label{tracequality11}
\\& =\tr(I) - \tr \left(( I+ (A+C)^{\frac{-1}{2}}B (A+C)^{\frac{-1}{2}})^{-1} \right)\nonumber
\\& \leq \tr(I) - \tr (I) \label{AprimeB11}
\\&= 0.\nonumber
\end{align}
in which the inequality \eqref{tracequality11} follows from the trace interchange property  and  \eqref{AprimeB11} is derived using  Lemma \ref{PeirelsLemma} for the Hermitian  matrix $(I+ (A+C)^{\frac{1}{2}}B (A+C)^{\frac{1}{2}})^{-1}$. Note that the matrix $(A+C)^{\frac{1}{2}}B (A+C)^{\frac{1}{2}}$ is a block off-diagonal matrix, since matrices $A$ and $B$ has been defined to be, respectively,  block diagonal and off-diagonal matrices such that wherever $A_{ij}$ is zero, $B_{ij}$ is non-zero and vice versa. 
This completes the proof of  Theorem \ref{Lemmaab}.
\end{proof}

\end{document}